\documentclass[preprint,authoryear,12pt]{elsarticle}

\usepackage{amsmath,amssymb,amsfonts}
\usepackage{graphicx}
\usepackage{booktabs}
\usepackage{hyperref}
\usepackage{geometry}
\usepackage{float}
\usepackage{bm}
\usepackage{algorithm}
\usepackage{algorithmic}
\usepackage{empheq}
\usepackage{xcolor}
\usepackage{subcaption}
\usepackage{graphicx}
\usepackage{comment}
\usepackage[T1]{fontenc}
\usepackage{amsthm} 
\theoremstyle{definition} 
\newtheorem*{corollary}{\normalfont\textit{Corollary}} 
\newtheorem*{remark}{\normalfont\textit{Remark}} 

\newtheorem*{mytheorem}{\normalfont\textit{Theorem}} 
\journal{European Journal of Operational Research}

\begin{document}

\begin{frontmatter}

\title{Temperature Anomalies and Climate Physical Risk\\ in Portfolio Construction}

\author[inst1]{Michele Azzone\corref{cor1}}
\ead{michele.azzone@polimi.it}
\cortext[cor1]{Corresponding author}

\author[inst1]{Carlo Bechi}
\ead{carlo.bechi@mail.polimi.it}

\author[inst2]{Gabriele Sbaiz}
\ead{gabriele.sbaiz@deams.units.it}

\affiliation[inst1]{organization={Department of Mathematics, Politecnico di Milano},
    addressline={Via Bonardi 9},
    city={Milano},
    postcode={20149},
    country={Italy}}

\affiliation[inst2]{organization={Department of Economics, Business, Mathematics and Statistics, University of Trieste},
    addressline={Via A. Valerio, 4/1},
    city={Trieste},
    postcode={34127},
    country={Italy}}

\begin{abstract}
Driven by the increasing frequency and intensity of natural disasters and chronic climate threats, we investigate the impact of physical climate risk on global equity portfolios. By employing a panel regression analysis on sectoral returns, we provide statistical evidence that extreme temperature events exert a negative effect on most sectors. We introduce two novel metrics based on these temperature anomalies, \textit{Climate Risk Exposure} and \textit{Climate Exposure Volatility}, in order to measure the environmental vulnerability of a portfolio. Unlike available static country-level indices, these metrics incorporate the time varying probability of extreme events and their relations with firm-specific asset intensity. We integrate these measures into a multi-objective portfolio optimization framework. This approach extends the traditional Mean-Variance paradigm, allowing investors to construct portfolios that are resilient to physical climate shocks without sacrificing diversification. Finally, we conduct a backtesting analysis to show the practical benefits of incorporating these climate risk metrics into the investment process, evaluating how climate-aware strategies perform relative to traditional benchmarks.
\end{abstract}


\end{frontmatter}

\section{Introduction}
	Physical climate risk is becoming a relevant factor in the asset management industry, comprising both acute hazards and chronic long-term shifts \citep{bressan2024asset}. Given the increasing frequency of climate-related phenomena\footnote{See \url{https://www.eea.europa.eu/en/analysis/indicators/economic-losses-from-climate-related}} and the continued rise in greenhouse gas emissions, asset managers and financial institutions worldwide are seeking methods to incorporate these risks into investment decisions \citep{campiglio2023climate,azzone2026physical,luciani2026transition}. Despite this growing attention, there remains a significant gap in effective quantitative metrics capable of capturing how corporate economic value is affected and, in turn, guiding the decisions of a rational investor \citep{bortolan2024volatile}.
    We propose two new metrics for the expected frequency and volatility of extreme climate events within a portfolio, based on temperature anomalies (see Section \ref{sec:measures_of_risk_for_extreme_climate_events} below). We then integrate these metrics into a multi-objective optimization problem, extending the mean-variance framework to account for physical risk. To the best of our knowledge, this is the first attempt to offer a flexible, data-accessible approach for constructing optimal portfolios that explicitly consider physical climate risk and its diversification.
    
	Climate risk has emerged as a critical financial threat, broadly categorized into physical and transition risks \citep{bolton2023global}. Focusing on the former, the last decade has witnessed a sharp increase in both acute hazards, such as heatwaves, floods, unprecedented wildfire activity, and chronic impacts like rising sea levels, with extreme events becoming up to 20 times more likely in certain regions \citep{jones2024state}.  Despite global initiatives aimed at mitigation, the escalating frequency of these physical threats poses severe financial consequences for firms that remain ill-prepared for the resulting economic losses \citep{wang2023climate}.
	
	Classifying physical climate risk for a portfolio is a complex endeavor, and consequently very few contributions are currently available. While some country-level indices exist, such as the Notre Dame Global Adaptation Initiative (ND-GAIN)\footnote{For further detail refer to: \url{https://gain.nd.edu/our-work/country-index/}}, they have significant limitations. The ND-GAIN index quantifies a country's exposure, sensitivity, and adaptive capacity across six critical sectors (food, water, health, ecosystem services, human habitat, and infrastructure). However, it does not provide a multivariate time-series of extreme events but just a country level score. A timeseries framework would be essential for understanding the impacts on stock returns and the correlations between events, which are necessary conditions for determining whether portfolio diversification is possible.
	Contributions in this framework usually focus on a timeseries model for a single risk, e.g., wildfires for mortgages credit risk in the case of \citet{kahn2024adaptation}. They show that Mortgage-Backed Securities (MBS) can effectively mitigate physical climate risk for retail mortgages by minimizing the spatial concentration and correlation of the underlying assets. By leveraging these diversification metrics, the authors construct optimal, climate-resilient portfolios that minimize financial losses. We aim to apply a similar framework to global equity portfolios, addressing exposure to general climate-related phenomena, which we proxy using extreme temperature anomalies \citep[see e.g.,][and references therein]{fierro2014relationships, tzouvanas2019can, attilio2025impact, ge2025abnormal}.
	
	There is a large literature on incorporating climate risk in asset management. However, most contributions focus on sustainability or transition risks, with very few addressing physical risks \citep[see][for a survey]{le2022portfolio}. In this context, the literature on ESG (Environmental, Social and Governance) portfolio selection generally follows two main approaches: incorporating ESG preferences directly into the utility function \citep{pedersen2021responsible, avramov2022sustainable,bertelli2025sustainable} or applying explicit constraints to the optimization problem \citep{de2023esg, morelli2024responsible, aprea2025neural}. Empirical studies utilizing the latter, specifically those minimizing variance or tracking error subject to ESG or carbon targets, outline that such constraints do not necessarily compromise risk-adjusted returns \citep{bolton2022net, azzone2026asset}.  A notable exception is the recent preprint by \citet{luciani2026transition}, which integrates a constraint—based on a risk measure constructed using the ND-GAIN index—into the traditional Mean-Variance framework. 
    
    In much of the aforementioned literature, the additional objective is either incorporated into the utility function—necessitating a specific  risk aversion parameter—or imposed as an arbitrary constraint. In contrast, we argue that generating a Pareto frontier via multi-objective optimization is preferable, as it provides decision-makers with a comprehensive map of the trade-offs between expected returns, standard financial risk, and physical climate risk. Moreover, minimizing physical risk variance is crucial because it accounts for physical risk diversification—a dynamic overlooked when merely aggregating individual asset scores as is done in the industry and by commercial providers \citep{luciani2026transition}. Furthermore, the inclusion of non-negativity constraints on portfolio weights renders analytical solutions intractable, making a numerical optimization approach necessary.

To tackle the multi-objective problem that considers physical risk together with standard mean variance, we employ a Multi-Objective Particle Swarm Optimization (MOPSO).
\citet{moore1999application} was the first to extend PSO \citep{kennedy1995particle} to multi-objective optimization, showing that swarm-based approaches are capable of effectively handling conflicting objectives and thereby establishing the basis for subsequent developments. Based on this paradigm, \citet{coello2004handling} proposed a MOPSO algorithm that integrates Pareto dominance with an external archive to preserve solution diversity and enhance convergence. Their work also includes a sensitivity analysis, providing valuable insights into parameter tuning for improved efficiency and diversity \citep[see][for a comprehensive review of MOPSO developments]{reyes2006multi}. \citet{durillo2009multi} performed a comparative study of existing MOPSO variants, identifying key challenges related to diversity maintenance. To address these issues, they introduced a MOPSO variant incorporating velocity constraint mechanisms, resulting in improved stability and diversity preservation. These findings highlight the necessity of adaptive control strategies in MOPSO algorithms.  We refer to \citet{caselli2025bilevel,ehrgott2025fifty,manzoni2026sustainable} for recent applications to sustainability.

	In this work, we bridge the gap between static country-level indicators and the need for dynamic, firm-specific physical risk assessment. Unlike existing literature that often relies on qualitative scores or focuses on specific asset classes like real estate, we provide a generalized framework applicable to global equity portfolios. First, we confirm the materiality of physical risk by empirically testing the impact of temperature anomalies on sector total returns. Second,  we model the time-varying probability of extreme events and define two novel metrics for climate physical risk that are based on physical risk correlations in different continents: the \textit{Climate Risk Exposure} and the \textit{Climate Exposure Volatility}.  Finally, we integrate these metrics directly into a multi-objective portfolio optimization problem. This allows us to extend the standard mean-variance analysis to explicitly minimize the volatility induced by physical climate shocks, offering a viable tool for climate-resilient portfolio construction and climate risk diversification.

The remainder of the paper is structured as follows. Section \ref{sec:extreme_temperature_anomalies} details the methodology used to construct the time series of temperature anomalies and evaluates how these anomalies influence equity returns.
Section \ref{sec:measures_of_risk_for_extreme_climate_events} proposes two distinct risk metrics to quantify physical climate risk. Section
\ref{sec:climate_opt} performs a portfolio optimization that incorporates these climate-sensitive measures. 
 Section \ref{sec:backtesting} conducts a backtesting analysis to evaluate how climate-aware strategies perform relative to traditional benchmarks, while Section \ref{sec:conclusion} concludes. The Appendices are devoted to reporting the proofs and to providing further details and results of the three-objective optimization framework.

\section{Temperature anomalies and their effects on returns}
\label{sec:extreme_temperature_anomalies}

This section details the methodology used to construct and model the time series of extreme temperature events. The goal is to estimate a time-varying probability of temperature anomalies and their correlations in different geographic regions, which serve as a fundamental input for the portfolio optimization framework. We choose to focus on temperature extremes as they are often the primary drivers of extreme climate events, such as wildfires, prolonged droughts, very intense heatwaves and glacial melt episodes \citep[see, e.g., ][]{fierro2014relationships}. Nevertheless, we believe that our approach and methodology are designed to be flexible, allowing them to be extended to any other extreme climate event under consideration.

We use temperature data provided by Our World in Data\footnote{See \url{https://ourworldindata.org}}. The platform sources information exclusively from official international institutions such as the IPCC, NASA, and the World Bank.
We utilize monthly mean temperature data for the period spanning from January 1, 1940, to April 30, 2025. The analysis covers the six major continents: North America, South America, Europe, Africa, Asia, and Oceania. 
Antarctica is excluded from this study. This decision is based on the assumption that the percentage of revenues generated by companies within the MSCI World index in this continent is negligible, thus having no significant impact on the financial risk-return profile of the portfolio.
We adopt a continent-level approach for temperature data because we have company revenue information available at this geographic scale. However, we believe that our analysis could be extended and rendered even more effective by utilizing temperature and revenue data at a national level. 
Such a detailed approach would allow for a more precise identification of the magnitude of climate risks to which the portfolio is exposed within each individual country.

\subsection{Identification of temperature anomalies}

To rigorously identify extreme temperature events, we implement a statistical detection method based on historical baselines.
First, we establish a reference period (baseline) spanning from 1960 to 1990. For each continent $k$ and each month $m \in \{1, \dots, 12\}$, we calculate the historical mean $\mu_{k,m}$ and the standard deviation $\sigma_{k,m}$ using only the observations within this 30-year window. This allows us to capture the seasonal cycle and the natural variability inherent to each specific month of the year and to each specific area of the globe. 

To ensure comparability between different continents and months, the raw monthly temperature $T_{k,t}$ observed at time $t$ is transformed into a standardized anomaly. The transformation is defined as
\begin{equation}
    Z_{k,t} = \frac{T_{k,t} - \mu_{k,m}}{\sigma_{k,m}}\;\;,
\end{equation}
where $m$ is the month corresponding to time $t$. This standardization allows for the consistent comparison of extreme events across diverse climatic zones, regardless of differences in baseline temperatures or local variances. 

We define an indicator variable $B_{k,t}$ to formally identify the occurrence of a positive temperature anomaly at time $t$ in continent $k$. An extreme event is recorded when the standardized temperature exceeds a critical threshold of 2 as follows: \begin{equation}
    B_{k,t} = 
    \begin{cases} 
    1 & \text{if } Z_{k,t} > 2 \\
    0 & \text{otherwise}
    \end{cases}\;\;.
\end{equation}
This threshold is statistically significant as, under the assumption of a normal distribution, it identifies events that lie in the extreme upper tail (approximately the top 2.28\% of observations under Gaussianity assumption) and ensures that we focus on ``extreme'' events.

To capture the dynamic evolution of climate risk over time, we estimate the probability $p_{k,t}$ that an anomaly occurs in continent $k$ at time $t$. Given the binary nature of $B_{k,t}$, we employ a logistic regression model, which is the standard approach for modeling the probability of occurrence of discrete events.

To reflect the complexities of global warming patterns, the model incorporates a quadratic time trend. This allows us to account for potential non-linearities and accelerations in temperature increases that a simple linear trend might fail to capture. For each continent, the log-odds of the probability are modeled as

\begin{equation}
    \label{eq:logit_model}
    \text{logit}(p_{k,t}) = \ln\left( \frac{p_{k,t}}{1 - p_{k,t}} \right) = \beta_{0,k} + \beta_{1,k} t + \beta_{2,k} t^2\;\;,
\end{equation}

where $t$ represents the time index, $\beta_{1,k}$ captures the linear progression of the risk over $t$, and $\beta_{2,k}$ accounts for the quadratic component, identifying whether the frequency of extreme events is accelerating or decelerating.

This specification ensures that our climate risk measures are forward-looking and capable of reflecting the intensifying nature of environmental threats, providing a more robust basis for the subsequent financial risk-return analysis.

As illustrated in Figure \ref{fig:temperature_anomalies}, the fitted probabilities $\hat{p}_{k,t}$ reveal a consistent trend across all analyzed continents.  The probability of encountering an extreme temperature anomaly exhibits a significant increase over the sample period.

\begin{figure}[H]
    \centering
    \includegraphics[width=0.95\textwidth]{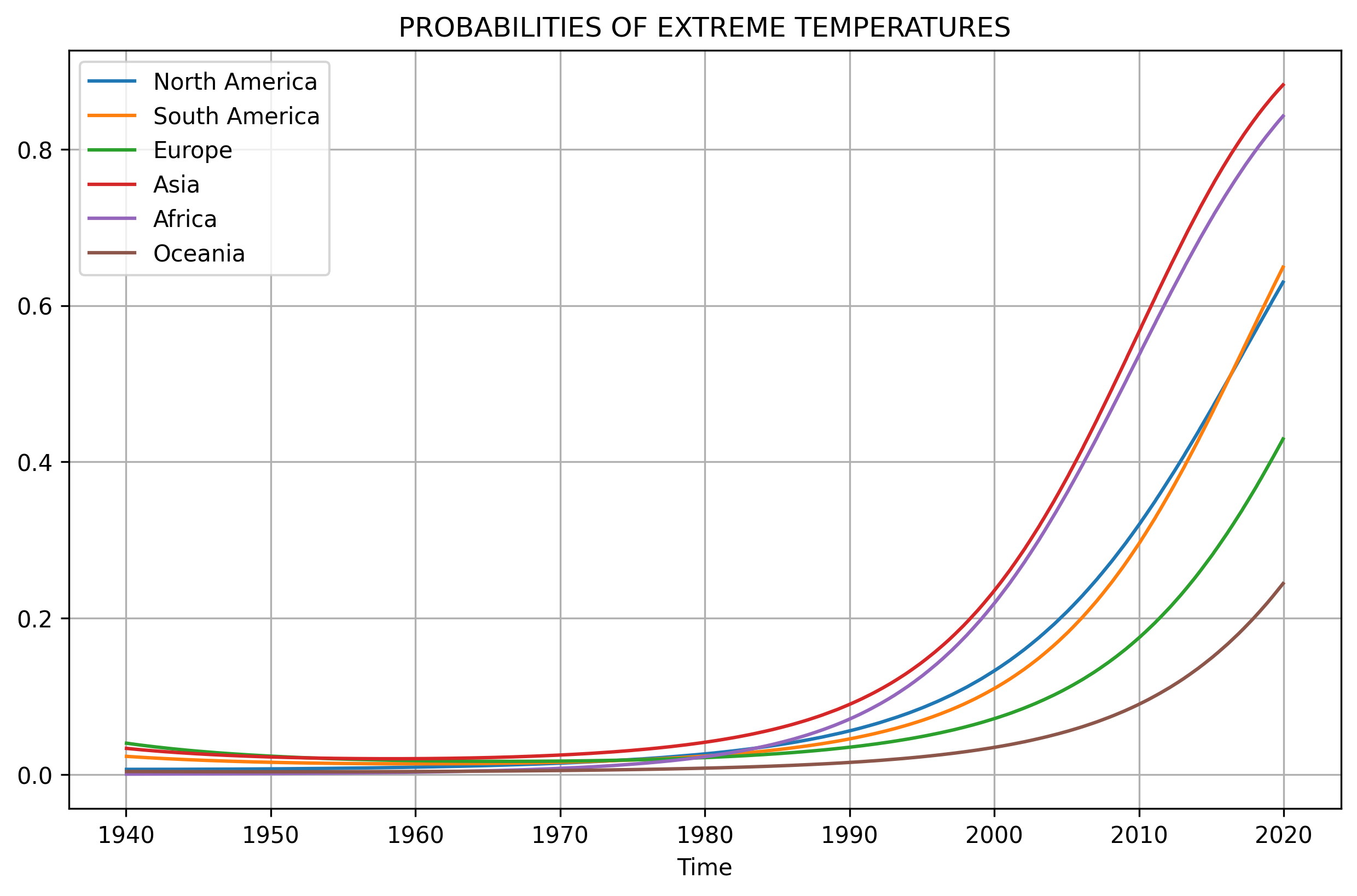} 
    \caption{Probabilities of extreme temperatures over time in the period 1940-2020.}
    \label{fig:temperature_anomalies}
\end{figure}

These pronounced upward trajectories provide empirical evidence that physical climate risk is non-stationary and has intensified in recent decades. In statistical terms, the traditional assumption of a constant risk distribution is invalidated by the clear intensification of temperature anomalies in recent decades.

The performance of the logistic regression models, evaluated through the Area Under the Receiver Operating Characteristic curve (AUROC) and Area Under the Precision-Recall curve (AUPR), is summarized in Table \ref{tab:logistic_performance}. To ensure a robust assessment of the models' reliability, the regression parameters are estimated using a historical training period from 1940 to 2025. 

To evaluate the models, we primarily focus on the AUROC. This metric measures the model's discriminatory power, and an AUROC of 0.5 suggests a predictive ability no better than a random guess, while a value of 1.0 represents perfect classification. As shown in the results, all continents exhibit exceptionally high AUROC values, ranging from 0.8380 for Europe to 0.9396 for Africa. 

Furthermore, we evaluate the model performance using the AUPR, which provides a more stringent evaluation in the presence of imbalanced datasets. 

In our analysis, the AUPR values for all continental models consistently and significantly outperform those of a baseline random classifier. We also have verified that very similar results are obtained when considering a logistic regression trained on the training sample (1940-2019) that we will utilize for portfolio selection. Figure \ref{fig:confronto_continenti} illustrates the probability of extreme events for Europe and North America, comparing forecasts derived from logistic regressions estimated over two distinct sample periods: 1940–2019 and 1940–2025. The fitted curve from the restricted training set closely tracks the curve estimated on the full sample. This alignment highlights that the relatively simple and stylized logistic regression model generalizes exceptionally well.
\begin{figure}[htbp]
    \centering
    \begin{subfigure}[b]{0.8\textwidth}
        \centering
        \includegraphics[width=\textwidth]{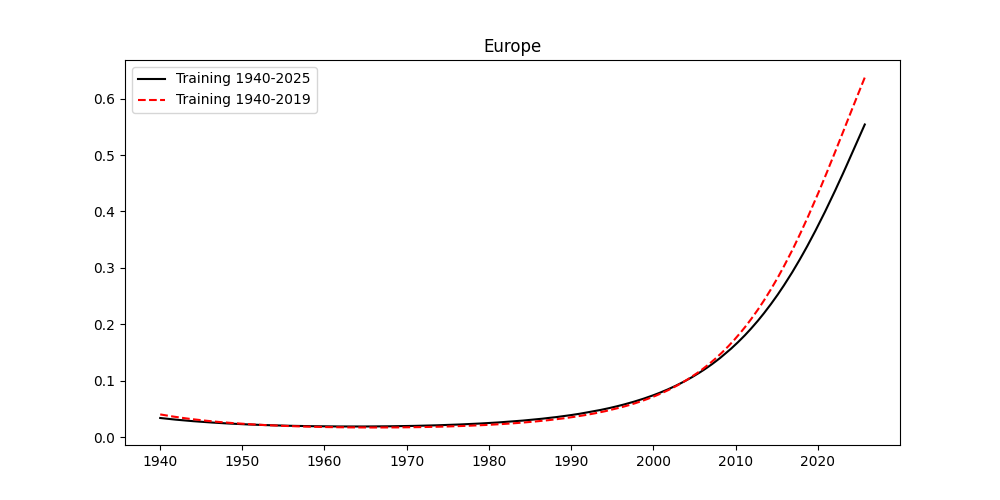}
        \label{fig:europe}
    \end{subfigure}
    \hfill
    \begin{subfigure}[b]{0.8\textwidth}
        \centering
        \includegraphics[width=\textwidth]{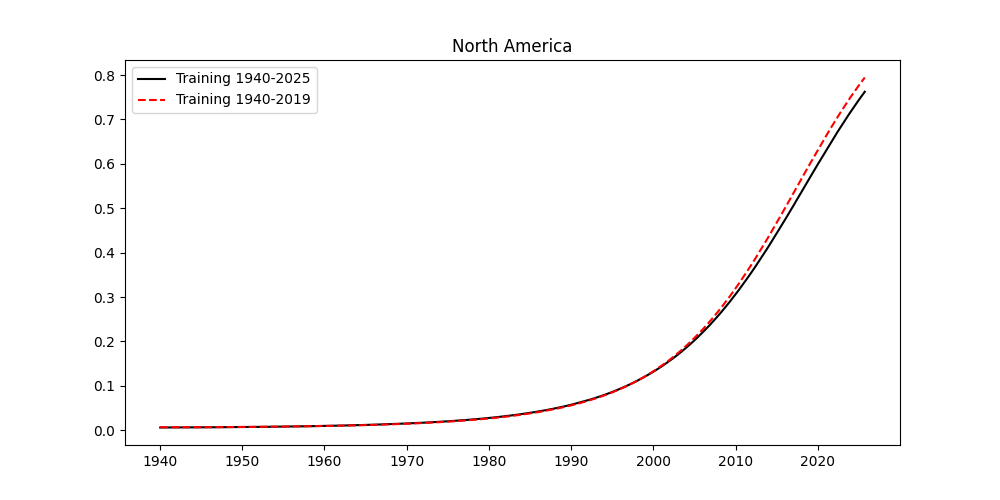}
        \label{fig:north_america}
    \end{subfigure}
    \caption{Estimated probability of extreme events for Europe and North America using logistic regression. The results compare models estimated over the 1940–2025 sample (continuous black line) and the 1940–2019 sample (dashed red line).}
    \label{fig:confronto_continenti}
\end{figure}

In conclusion, the high predictive scores across all geographical regions allow for the categorical rejection of a stationary climate hypothesis. This validates the use of these time-varying probabilities as robust inputs for the subsequent multi-objective portfolio optimization.

\begin{table}[H]
\centering
\caption{Performance metrics for logistic regressions. }
\label{tab:logistic_performance}
\begin{tabular}{lccc}
\hline
\textbf{Continent} & \textbf{AUROC} & \textbf{AUPR} & \textbf{AUPR Random} \\ \hline
Africa          & 0.9396 & 0.8031 & 0.2060 \\
Asia            & 0.9067 & 0.7312 & 0.2225 \\
Europe          & 0.8380 & 0.3776 & 0.0943 \\
North America   & 0.8947 & 0.6100 & 0.1429 \\
Oceania         & 0.8830 & 0.3215 & 0.0564 \\
South America   & 0.8754 & 0.5227 & 0.1341 \\ \hline
\end{tabular}
\end{table}


\subsection{Correlation analysis and spatial dependence}

After estimating the time-varying probabilities of temperature anomalies for each continent, it is essential to analyze the spatial dependence between these events. The objective of this analysis is to determine the correlation structure of climatic shocks across different geographic areas. 
In particular, we aim to quantify the risk that a global portfolio might be simultaneously affected by multiple extreme events occurring in different parts of the world.
We focus on estimating this correlation structure and verifying its theoretical consistency. Specifically, we examine the empirical correlations between the climate indicators and test whether the resulting correlation matrix satisfies the mathematical constraints inherent to a system of Bernoulli random variables. This verification is fundamental to ensuring that our model provides a consistent representation of global climatic interdependencies.

For each pair of continents $(i, j)$, we compute the Pearson correlation coefficient between their respective binary time series of anomalies, $B_{i,t}$ and $B_{j,t}$.
We are able to visualize the result in the following empirical correlation matrix, Figure~\ref{fig:correlation_matrix}.

\begin{figure}[H]
    \centering
    \includegraphics[width=0.85\textwidth]{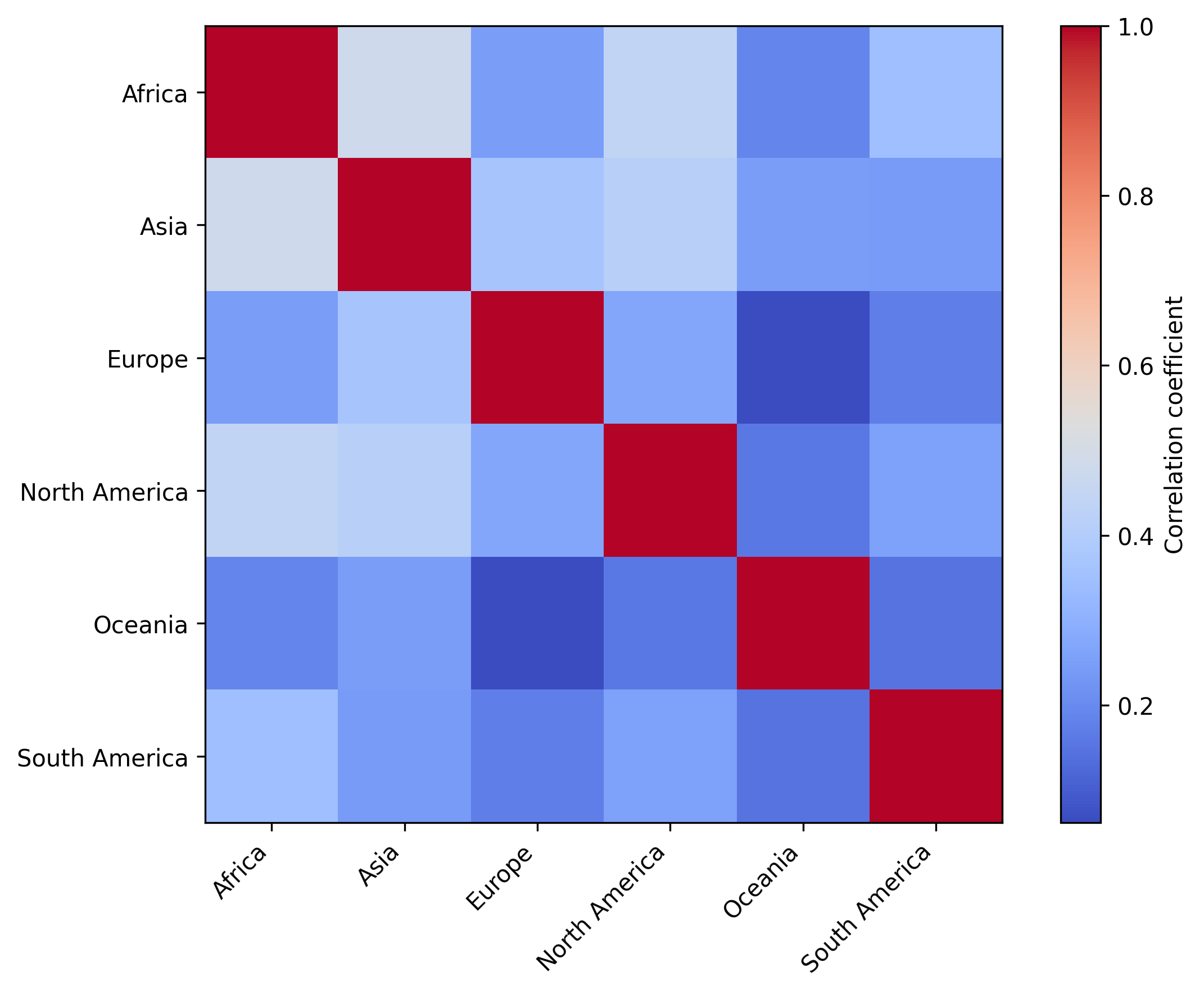} 
    \caption{Correlation matrix of extreme temperature time series.}
    \label{fig:correlation_matrix}
\end{figure}

A primary observation is that all estimated Pearson coefficients are positive. This positivity reflects a global systemic trend in climate anomalies, suggesting that temperature extreme events are driven by a common underlying process, global warming, that tends to lift temperatures across the entire planet simultaneously.

However, the magnitude of these correlations remains relatively low, with values ranging between 0 and 0.48. This evidence is particularly important for two reasons.  On one hand, the positivity confirms that no continent is truly immune or ``negatively correlated'' to the global warming trend. On the other hand, the moderate magnitude (below 0.48) suggests that a good degree of geographical diversification is still possible.

In conclusion, the heatmap reveals a world where climate risk is globally interconnected but still characterized by significant regional specificities.

\subsection{Fréchet-Hoeffding bounds}

While Pearson correlation is a standard measure for
Gaussian random variables, its application to Bernoulli distributions requires careful consideration.
Since our indicator variables $B_{k,t}$ are discrete and binary, the correlation coefficients are subject to mathematical bounds determined by the underlying probabilities of each event.  These bounds could be very interesting for developing robustness and stress test analysis on climate shock coordination.

Every pair of Bernoulli random variables has a maximum and minimum admissible values of correlation  (\citealt{huber2015multivariate}), this constraint is formally expressed by the Fréchet-Hoeffding theorem: 

\begin{mytheorem}[Fréchet-Hoeffding]
For any two random variables $X_1$ and $X_2$ with cumulative distribution functions $F_1$ and $F_2$, with $U \sim \text{Unif}([0, 1])$ the following inequality holds:

\begin{equation}
    \text{corr}(F_1^{-1}(U), F_2^{-1}(1-U)) \leq \text{corr}(X_1, X_2) \leq \text{corr}(F_1^{-1}(U), F_2^{-1}(U))\, .
\end{equation}
\end{mytheorem}
We prove the corollary that follows for the specific case of two Bernoulli variables, $X_1 \sim \text{Bern}(p_1)$ and $X_2 \sim \text{Bern}(p_2)$ with different probability parameters.
\begin{corollary}
\label{cor:bernoulli_bounds}
The correlation coefficient $\rho$ between two Bernoulli random variables $X_1 \sim \text{Bern}(p_1)$ and $X_2 \sim \text{Bern}(p_2)$ is bounded by $\rho_{min} \le \rho \le \rho_{max}$, where:
\begin{equation}
    \rho_{min} = \frac{\max(0, p_1 + p_2 - 1) - p_1 p_2}{\sqrt{p_1(1-p_1)} \sqrt{p_2(1-p_2)}}
\end{equation}
\begin{equation}
    \rho_{max} = \frac{\min(p_1, p_2) - p_1 p_2}{\sqrt{p_1(1-p_1)} \sqrt{p_2(1-p_2)}} \;\;.
\end{equation}
\end{corollary}
\vspace{0.7cm}
\begin{proof}
The result is derived by applying the Fréchet–Hoeffding theorem, which establishes the existence of upper and lower bounds for joint distribution functions. 
To determine these limits, we analytically compute the covariance using the cumulative distribution functions of the two Bernoulli variables. By exploiting the properties of Bernoulli distributions, the terms $\text{Cov}(F_1^{-1}(u), F_2^{-1}(u))$ and $\text{Cov}(F_1^{-1}(u), F_2^{-1}(1-u))$ can be solved in closed form. Finally, by dividing these analytical covariances by the product of the standard deviations, $\sqrt{p_1(1-p_1)} \sqrt{p_2(1-p_2)}$, we obtain the explicit expressions for $\rho_{min}$ and $\rho_{max}$. For the full proof, please refer to \ref{app:corr_bound_proof}. \end{proof}
These bounds imply that a perfect correlation ($\rho = 1$) is only attainable when the marginal probabilities are equal ($p_1 = p_2$). Conversely, a perfect anti-correlation ($\rho = -1$) can only be achieved if the probabilities are complementary, such that $p_1 = 1 - p_2$. If these conditions are not met, the theoretical maximum and minimum correlations will strictly lie within the interval $(-1, 1)$.

In our framework, the parameters $p_{k,t}$ are non-stationary, as they are derived from a time-dependent logistic regression. This time-varying nature makes the direct application of the Fréchet-Hoeffding formulas complex, as the bounds would technically shift at every time step $t$.

To provide an estimate of the correlation boundaries, we decide to use the mean value of the estimated parameters $\hat{p}_{k,t}$ over the training horizon (1940--2020) as the representative $p_1$ and $p_2$ for the boundary calculations:
\begin{equation}
    \bar{p}_k = \frac{1}{T} \sum_{t=1}^T \hat{p}_{k,t}\;\;.
\end{equation}
This approach allows us to verify if the observed empirical correlations between continents are consistent with the theoretical limits imposed by their average probability of occurrence.

The results of this analysis are summarized in Figure~\ref{fig:correlation_boundaries}, which illustrates the empirical Pearson correlation for each pair of continents plotted against their respective Fréchet-Hoeffding boundaries. This visualization confirms that the empirical estimates fall strictly within the feasible theoretical range.

\begin{figure}[H]
    \centering
    \includegraphics[width=0.9\textwidth, keepaspectratio]{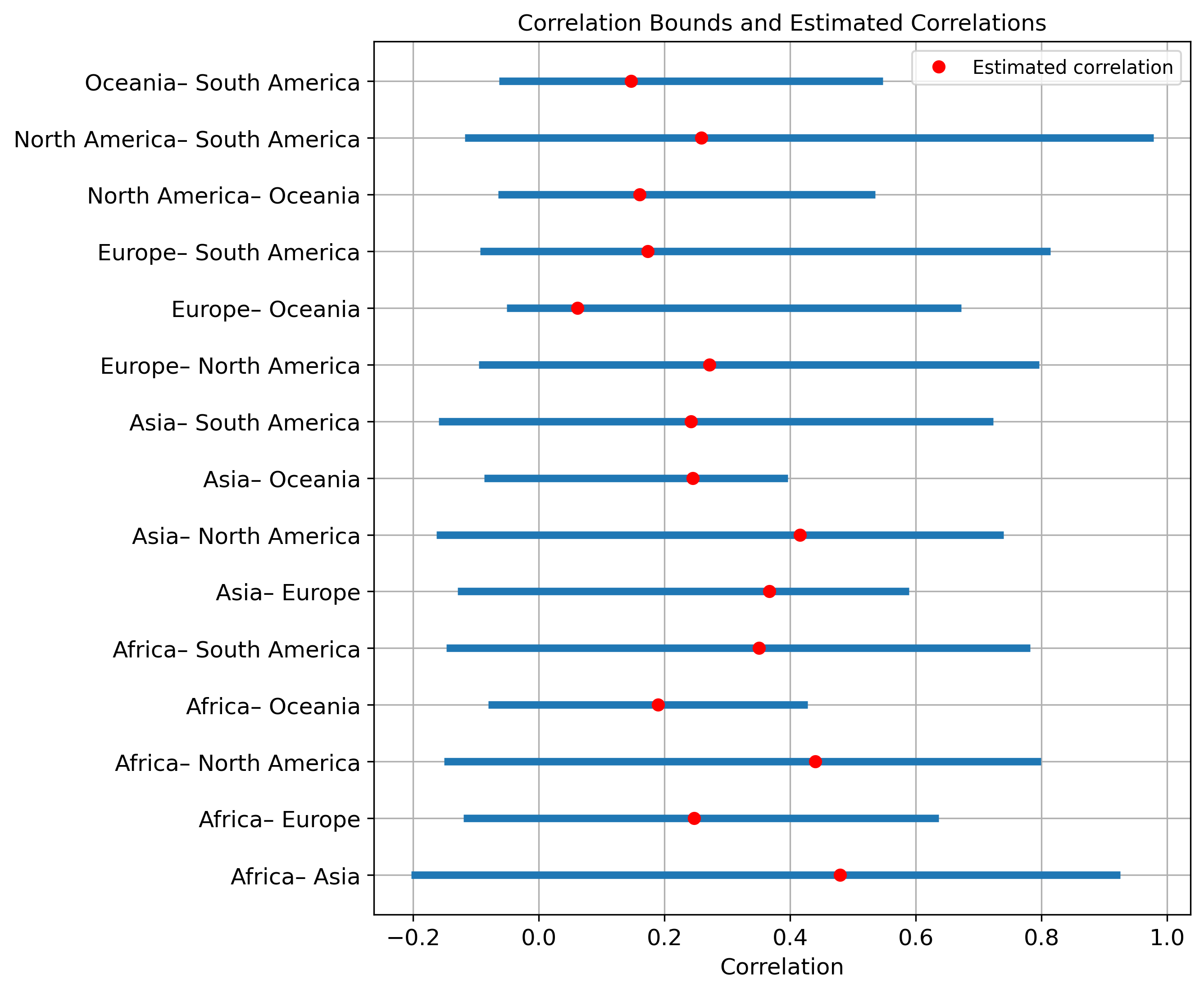}
    \caption{Correlation boundaries of extreme temperature time series.}
    \label{fig:correlation_boundaries}
\end{figure}

It is well-established that when extending to a multivariate Bernoulli distribution involving more than two variables, additional constraints must be satisfied to define the range of admissible correlations \citep{huber2015multivariate, huber2019admissible}. However, these higher-order constraints are generally not explicit and lack a closed-form expression, particularly for systems exceeding four variables. Since our system comprises six variables, one for each continent, applying such higher-order constraints becomes analytically intractable. Consequently, to the best of our knowledge, the pairwise Fréchet-Hoeffding bounds represent the most robust analytical tool available.  


\subsection{The impact of temperature extremes on equity returns}
\label{sec:impact_of_climate_extremes_on_equity_returns}

We now investigate the empirical relationship between extreme temperature anomalies and financial performance across various economic sectors. The core objective is to determine whether climate physical risks, manifested through continental temperature extremes, exert a statistically significant influence on asset returns via a linear regression model that includes climate extreme events among its covariates. The connection between extreme temperatures and returns is investigated by \citet{bortolan2024volatile}, where evidence is provided that high temperature anomaly variance, exhibits a statistically significant negative relationship with equity earnings. In this paper, instead, we try to understand whether extreme positive temperature anomalies, that can be interpreted as proxy of climate extreme events, affect stock returns. The analysis conducted in this section justifies the development of metrics designed to quantify the climate risk of a portfolio and the implementation of portfolio optimization strategies that explicitly incorporate the climate-related measures we propose. 

We construct sectoral return time series using the constituents of the MSCI World index. We consider the companies belonging to the MSCI World universe, categorized by their respective sector. Specifically, we select only those firms  provided with sufficient available data for the analysis and that were constituents of the MSCI World index at the beginning of 2020 and 2025. This selection criterion is essential to mitigate survivorship bias and ensure a consistent sample. 

We utilize total returns, which provide a comprehensive measure of performance by accounting not only for price changes but also for dividends, stock splits, and other corporate actions.\\
For each sector $s$, the return $R_{s,t}$ is calculated as the market capitalization weighted average of the individual total returns of its constituent firms within the MSCI World.

The foundational step of our analysis involves extending the traditional Capital Asset Pricing Model \citep{sharpe1964capital}.  We augment the framework by introducing a climate-risk proxy to capture the potential impact of physical climate shocks, specifically, we explicitly incorporate into the model, as covariates, the time series of extreme temperature anomalies previously constructed for each continent. This approach allows us to directly quantify the sensitivity of asset returns to localized climatic extremes.

For every sector $s$ and every continent $k$, we estimate the following regression:
\begin{equation}
    R_{s,k,t} - R_{f,t} = \alpha_{s,k} + \beta_{MKT,s,k} (R_{MKT,t} - R_{f,t}) + \gamma_{s,k} B_{k,t} + \epsilon_{s,k,t}\;\;,
\end{equation}

where $B_{k,t}$ is the temperature anomaly indicator for continent $k$, and $\gamma_{s,k}$ is the parameter of interest, representing  the sensitivity of the sector to extreme temperature events in a specific continent. Moreover, the quantity
$R_{s,k,t} - R_{f,t}$ represents the sectoral excess return, calculated by subtracting the monthly risk-free rate from the sectoral total return, and $R_{MKT,t} - R_{f,t}$ is the market excess return, derived from the performance of the overall MSCI World Index minus the risk-free rate.

The model parameters are estimated over the 2004-2025 timeframe, which corresponds to the period for which financial return data are available. To ensure the robustness of our statistical inference, the regression incorporates HAC (Heteroskedasticity and Autocorrelation Consistent) robust standard errors \citep{newey1986simple}. Furthermore, we exclude sectors regarding banks, financial and insurance companies from our analysis. This exclusion is justified by the fact that, without an in-depth analysis of their balance sheets and underlying portfolios, it is particularly challenging to determine their precise geographical exposure and the extent of their vulnerability to climate physical risks. 

The results are presented in the following Table \ref{tab:sector_continent_gamma}, specifically, only the $\gamma$ that are statistically significant are reported.

\begin{table}[H]
\centering
\caption{The table reports the estimated $\gamma_{s,k}$ coefficients and p-values for pairs reaching statistical significance ($p < 0.1$),  displaying only the coefficients that are statistically significant.}
\label{tab:sector_continent_gamma}
\small
\begin{tabular}{llcc}
\hline
\textbf{Sector} & \textbf{Continent} & \textbf{$\bf \boldsymbol{\gamma}_{s,k}$} & \textbf{p-value} \\ \hline
Applied Resources & North America & -0.0140 & 0.001*** \\
                           & Oceania       & -0.0110 & 0.049** \\ \hline
Automobiles \& Auto Parts & Asia       & -0.0138 & 0.052* \\
                             & Oceania     & -0.0183 & 0.100* \\ \hline
Chemicals & Africa               & -0.0064 & 0.017** \\
                   & Asia                 & -0.0076 & 0.003*** \\
                   & North America        & -0.0063 & 0.058* \\
                   & Oceania              & -0.0119 & 0.001*** \\ \hline
Consumer Goods Conglomerates & Africa  & -0.0097 & 0.053* \\
                                & Asia    & -0.0106 & 0.031** \\
                                & Oceania & -0.0094 & 0.067* \\ \hline
Cyclical Consumer Products & Asia  & -0.0075 & 0.068* \\
                                & Europe    & -0.0088 & 0.057* \\
                                & North America & -0.0083 & 0.069* \\
                                & Oceania & -0.0113 & 0.053* \\
                                \hline  
Cyclical Consumer Services & Asia & -0.0109 & 0.013** \\
                               & Oceania       & -0.0120 & 0.030** \\ \hline
Energy - Fossil Fuels & North America & -0.0153 & 0.063* \\
                               & Oceania       & -0.0192 & 0.066* \\ \hline
Food \& Beverages & Africa            & -0.0058 & 0.039** \\
                           & North America     & -0.0095 & 0.006*** \\
                           & Oceania           & -0.0115 & 0.010*** \\ \hline
Food \& Drug Retailing & North America &  -0.0098 & 0.020** \\
                           & Oceania           & -0.0135 & 0.015** \\ \hline                    
Industrial \& Commercial Services & Africa       & -0.0080 & 0.023** \\
                                & Asia         & -0.0139 & 0.000*** \\
                                & North America & -0.0096 & 0.001*** \\
                                & Oceania      & -0.0084 & 0.034** \\
                                & South America & 0.0068  & 0.058* \\ \hline
Mineral Resources & Africa            & -0.0136 & 0.028** \\
                           & Asia              & -0.0148 & 0.011** \\
                           & North America     & -0.0114 & 0.091* \\ \hline
Personal \& Household Products \& Services & North America        & -0.0066 & 0.088* \\ \hline       
Pharmaceuticals \& Medical Research & North America        & -0.0090 & 0.085* \\ \hline 
Real Estate & Africa       & -0.0077 & 0.070* \\ \hline
Technology Equipment & Asia           & -0.0077 & 0.016** \\
                              & Europe         & -0.0053 & 0.099* \\ \hline
Telecommunications Services & Oceania      & -0.0133 & 0.008*** \\ \hline
Transportation & North America        & -0.0068 & 0.026** \\ \hline
\multicolumn{4}{l}{\small \textit{Significance levels: *p$<$0.1, **p$<$0.05, ***p$<$0.01}} \\
\end{tabular}
\end{table}

The estimation of the model across all sector-continent pairs yields several noteworthy results. First, 17 out of 23 sectors exhibit returns that are significantly impacted by extreme temperature events in at least one continent. Second, a remarkably consistent pattern emerges among the statistically significant coefficients: the sign of $\gamma_{s,k}$ is always negative, with only a single exception. This predominant negative relationship suggests that extreme temperature anomalies generally exert a downward pressure on asset returns.

Furthermore, we observe that sectors such as Chemicals, Consumer Goods Conglomerates, Cyclical Consumer Products, Food \& Beverages, Industrial \& Commercial Services, and Mineral Resources are negatively impacted by extreme temperature events in at least three different continents. 
The heavy reliance on large-scale production plants and refineries in the Chemicals and Mineral Resources sectors, alongside extensive warehousing and distribution networks for Consumer Goods Conglomerates, Cyclical Consumer Products, and Food \& Beverages, creates significant vulnerabilities to climatic shocks. Specifically, extreme events can trigger severe operational disruptions. Rather than focusing solely on direct structural destruction, the heightened frequency of climate extremes leads to critical operational downtime, supply chain bottlenecks, and the potential degradation of temperature-sensitive physical inventories and assets \citep{starr2000effects, graff2014temperature, neidell2021temperature}.

While the initial findings offer a valuable preliminary perspective, they may be constrained by limited statistical power or omitted variable bias. To test the robustness of our results, we  estimate a specific $\gamma$ for each sector in a panel dataset framework by incorporating fixed effects to absorb systematic variations independent of temperature extremes.
 The panel regression reads as 

\begin{equation}
    R_{s,k,t} - R_{f,t} = \alpha_s + \beta_{MKT,s} (R_{MKT,t} - R_{f,t}) + \gamma_s B_{k,t} + \text{FE}_{\text{month}} + \text{FE}_{\text{continent}}+ \epsilon_{s,k,t}\;\;,
\end{equation}
with $\gamma_s$ the  climate coefficient for sector $s$, estimated across all continents. We consider continent fixed effects (FE), which control for time-invariant characteristics of each continent (e.g., geographic location or structural economic differences), and  month fixed effects, given by 12 dummy variables (one for each month of the year), which control for seasonality in both climate patterns and economic activity. 

This regression captures the common sensitivity of a sector's returns to extreme temperature shocks ($B_{k,t}$) across different regions. 


Table \ref{tab:panel_sector_results} presents the sectors characterized by a statistically significant $\gamma_s$. While most of these sectors are already introduced in Table \ref{tab:sector_continent_gamma}, this updated selection also includes Retailers and Software \& IT Services. 
Even under this specification, many sectors continue to exhibit statistically significant negative coefficients, with the exception of the Software \& IT Services sector which has historically maintained minimal physical exposure prior to the recent data center expansion. These findings confirm that the negative impact of climate extremes remains persistent across diverse industries. Moreover, it is worth mentioning that the sectors included in Table \ref{tab:sector_continent_gamma} but omitted from Table \ref{tab:panel_sector_results} continue to display a negative $\gamma_s$, although these estimates are no longer statistically significant.

\begin{table}[H]
\centering
\caption{The table reports the estimated $\gamma_s$ and the respective p-values, displaying only the coefficients that are statistically significant.}
\label{tab:panel_sector_results}
\begin{tabular}{lcc}
\hline
\textbf{Sector} & \textbf{$\bf \boldsymbol{\gamma}_s$} & \textbf{p-value} \\ \hline
Chemicals                        & -0.0062 & 0.0014*** \\
Energy - Fossil Fuels            & -0.0089 & 0.0399** \\
Food \& Beverages                & -0.0050 & 0.0070*** \\
Real Estate                      & -0.0055 & 0.0121** \\
Retailers                        & -0.0026 & 0.0811*   \\
Software \& IT Services          & 0.0026  & 0.0609* \\
Telecommunications Services      & -0.0042 & 0.0718* \\
Transportation                   & -0.0032 & 0.0893* \\ \hline
\multicolumn{3}{l}{\small \textit{Significance levels: *p$<$0.1, **p$<$0.05, ***p$<$0.01}} \\
\end{tabular}
\end{table}

As a final test of the relationship between sectoral returns and extreme climate events, we estimate a single global $\Gamma$ by pooling all sectors and all continents into a single panel model as
\begin{equation}
    R_{s,k,t} - R_{f,t} = \alpha + \beta_{\text{MKT}} (R_{\text{MKT},t} - R_{f,t}) + \Gamma B_{k,t} + \text{FE}_{\text{month}} + \text{FE}_{\text{continent}} + \epsilon_{s,k,t}.
\end{equation}

The results of this panel regression yield a coefficient $\Gamma$ that is negative and highly statistically significant, with an estimated value of $-0.003$ and a $p$-value of 0.0096. The fact that this result holds, while including both continent and month fixed effects, provides strong evidence that extreme temperature events represent a systematic physical risk factor. These climate shocks negatively affect firm valuations globally, justifying the inclusion of climate risk exposure in the portfolio optimization process. 


\section{Climate Physical risk metrics}
\label{sec:measures_of_risk_for_extreme_climate_events}

This section outlines the construction of two distinct metrics designed to quantify physical climate risk within a financial portfolio. The objective is to measure the impact of extreme weather events through two specific dimensions: the \textit{Climate Risk Exposure} (CRE), which captures the average expected exposure, and the \textit{Climate Exposure Volatility} (CEV), which quantifies the instability and volatility of that exposure. 

As a proxy for extreme climate events, this analysis utilizes the previously constructed time series of extreme temperatures for each continent \citep{tzouvanas2019can, attilio2025impact}. Temperature extremes are the fundamental drivers behind a vast array of physical disasters, including heatwaves, wildfires, droughts, intense storms, and hurricanes \citep{fierro2014relationships}.

While this study focuses on temperature as a primary indicator, the framework is designed to be modular. The methodology can be easily extended to more specific climate hazards (e.g., flood risk, sea-level rise) or refined to a more detailed geographic scale, such as country-by-country analysis, provided the relevant Bernoulli time series and the geographical exposure of the firms are available.

\subsection{Firms' Vulnerability}
The sensitivity of a firm to physical climate risk depends on the geographic distribution of its revenues and the vulnerability of its physical assets and supply chain in those areas. Consequently, an industrial company with significant physical infrastructure is more vulnerable to a local climate disaster than a service-based firm. As in \citet{azzone2026physical}, we define asset intensity of firm $i$ as 
\begin{equation}
AI_i = \frac{\text{Tangible Assets}_i}{\text{Revenue}_i}\;\;.
\end{equation}

We believe this metric is highly informative, as it reflects the degree of a company’s dependency on physical assets to generate revenue.  When aggregating this metric by sector, our empirical results confirm our economic intuition: the Real Estate sector exhibits the highest average asset intensity ($10.6$), reflecting a business model where value is primarily linked to the ownership of physical structures and land. Conversely, the Retailers sector shows the lowest average asset intensity ($0.6$), consistent with operations characterized by high inventory turnover and a lower reliance on fixed tangible assets relative to total sales. This polarization highlights that physical vulnerability is not uniform across the economy but is strictly linked to the underlying asset composition. 

To provide a geographical localization of the risk, we also define $S_{i,k}$ the revenue produced in continent $k$ by firm $i$ and by combining it with the asset intensity, we determine the importance of continent $k$ for company $i$ in term of physical risk. 
Furthermore, we define the climate normalized portfolio weight for continent $k$, denoted as $\alpha_k(\bm{w})$,
\begin{equation}
\label{eq:alpha_calc}
\alpha_k(\bm{w}) = \frac{\sum_{i=1}^{N} w_i \cdot AI_i \cdot S_{i,k}}{\sum_{k=1}^K \sum_{i=1}^{N} w_i \cdot AI_i \cdot S_{i,k}} \;\;,
\end{equation}
where $w_i \in \mathbb{R}$ is the portfolio weight of firm $i$ and $N$ is the number of stocks. By using these climate normalized weights we are able to capture the physical risk exposure of the portfolio to a given continent. \\
Unsurprisingly, calculating the climate normalized weights for the MSCI World index, Asia, North America, and Europe exhibit the highest climate normalized weights. In particular, Asia shows the most significant exposure. This is motivated by its role as the ``world's factory'', characterized by a dense concentration of manufacturing hubs and industrial facilities that inherently increase the physical asset intensity of the region.

Using the climate normalized weights and the temperature extreme time series $B_{k,t}$, we can define our two new physical risk metrics. 

The first measure, \textit{Climate Risk Exposure}, quantifies the expected average exposure of the portfolio to extreme climate events at a given time $t$. It represents the ``baseline'' physical risk, calculated as the expected value of the weighted sum of the time-varying Bernoulli variables $B_{k,t}$, using the climate normalized weights $\alpha_k(\bm{w})$. Unlike static indicators, this measure accounts for the non-stationarity of climate trends by utilizing the time-dependent probabilities $p_{k,t}$. It is defined as
\begin{equation} \label{eq:cre_final}
\text{CRE}(w, t) = \mathbb{E} \left[ \sum_{k=1}^K \alpha_k(w) B_{k,t} \right] = \sum_{k=1}^K \alpha_k(w) \cdot p_{k,t}\;\;,
\end{equation}
where $p_{k,t}$ is the probability of an extreme temperature event in continent $k$ at time $t$. A higher $\text{CRE}$ indicates that the portfolio is potentially highly affected by regions currently experiencing a high frequency of climate anomalies. 

While the $\text{CRE}$ provides a mean expectation of physical risk, it inherently lacks the capacity to account for the uncertainty or the potential for simultaneous disasters across different geographic areas. To address this limitation, we introduce the \textit{Climate Exposure Volatility}, denoted as $\text{CEV}(\bm{w}, t)$. This metric quantifies the variance of the portfolio's climate exposure at time $t$, offering a rigorous measure of the stability of its risk profile. A high volatility suggests that the portfolio's risk is poorly diversified across climate zones. In such cases, the portfolio becomes highly susceptible to sudden, large-scale shocks that affect just one geographical zone. Formally, the function $\text{CEV}$ is defined as the variance of the weighted sum of the Bernoulli variables $B_{k,t}$, utilizing the climate normalized weights $\alpha_k(\bm{w})$ previously derived. It is given by

\begin{equation}
\label{eq:f3_final}
\text{CEV}(\bm{w}, t) = \text{Var} \left( \sum_{k=1}^K \alpha_k(w) B_{k,t} \right) = \sum_{k=1}^K \alpha_k^2 \text{Var}(B_{k,t}) + 2 \sum_{k < j} \alpha_k \alpha_j \text{Cov}(B_{k,t}, B_{j,t})\, ,
\end{equation}

where $\text{Var}(B_{k,t}) = p_{k,t}(1-p_{k,t})$ is the variance of the Bernoulli process for continent $k$ at time $t$ and $\text{Cov}(B_{k,t}, B_{j,t})$ is the covariance between the events in continents $k$ and $j$.  Formally, the covariance is expressed as

\begin{equation}
\text{Cov}(B_{k,t}, B_{j,t}) = \rho_{k,j} \sqrt{\text{Var}(B_{k,t}) \text{Var}(B_{j,t})} = \rho_{k,j} \sqrt{p_{k,t}(1-p_{k,t}) p_{j,t}(1-p_{j,t})}\;\;.
\end{equation}

Substituting these expressions we can rewrite the \textit{Climate Exposure Volatility} as
\begin{equation} \label{eq:f3_explicit} \text{CEV}(\bm{w}, t) = \sum_{k=1}^K \alpha_k^2 p_{k,t}(1-p_{k,t}) + 2 \sum_{k < j} \alpha_k \alpha_j  \rho_{k,j} \sqrt{p_{k,t}(1-p_{k,t}) p_{j,t}(1-p_{j,t})} \;\;.\end{equation}

This formulation allows for a decomposition similar to the one found in \citet{kahn2024adaptation}, that decomposes the wildfire risk in the context of a portfolio of Mortgage-Backed Securities (MBS). The first term, $\sum_{k=1}^K \alpha_k^2 p_{k,t}(1-p_{k,t})$, is akin to a Herfindahl Index and measures the geographic dispersion of risk. If we assume uniform probabilities $p_k$, this term is minimized when the portfolio's climate exposure is equally distributed across all $K$ regions (i.e., $\alpha_k = 1/K$ for all $k$). The second term, $2 \sum_{k < j} \alpha_k \alpha_j \text{Cov}(B_{k,t}, B_{j,t})$, accounts for the systemic component of climate risk, reflecting how the correlation of extreme events across different geographic areas amplifies the overall portfolio volatility: if extreme events in two continents (e.g., North America and Europe) are highly correlated, a portfolio with significant exposure to both will exhibit higher volatility. This metric effectively penalizes a lack of geographic diversification against synchronized climate shocks, providing a more comprehensive view of potential portfolio losses than a simple average exposure metric. 

Furthermore, it is insightful to examine the behavior of this metric as the number of geographic regions $K$ increases (e.g., shifting from a continental to a country-level allocation). Let us consider a simplified scenario where the portfolio's climate exposure is uniformly distributed across all regions, such that $\alpha_k = 1/K$ for all $k$, and assume a constant risk level $p_k = p$ and a uniform pairwise correlation $\rho_{k,j} = \rho$. Under these assumptions, the first term of Equation \eqref{eq:f3_explicit} becomes
\begin{equation}
    \sum_{k=1}^K \left(\frac{1}{K}\right)^2 p(1-p) = \frac{1}{K} p(1-p)\, .
\end{equation}
Similarly to what was demonstrated in \citet{kahn2024adaptation}, as $K \to \infty$, this idiosyncratic component tends to zero, reflecting the benefit of granular geographic diversification. However, the second term (the systemic component) behaves differently. Since there are $K(K-1)/2$ pairs in the double summation, the correlation term simplifies to
\begin{equation}
    2 \sum_{k < j} \frac{1}{K^2} \rho \, p(1-p) = 2 \frac{K(K-1)}{2} \frac{1}{K^2} \rho \, p(1-p) = \left(1 - \frac{1}{K}\right) \rho \, p(1-p)\;\;.
\end{equation}
As $K \to \infty$, this term converges to the constant value $\rho \, p(1-p)$. This result highlights that while increasing geographic granularity can eliminate local risk dispersion, it cannot completely eliminate the systemic risk arising from the underlying correlation structure of climate events.

\begin{figure}[H]
    \centering
    \includegraphics[width=1.0\textwidth]{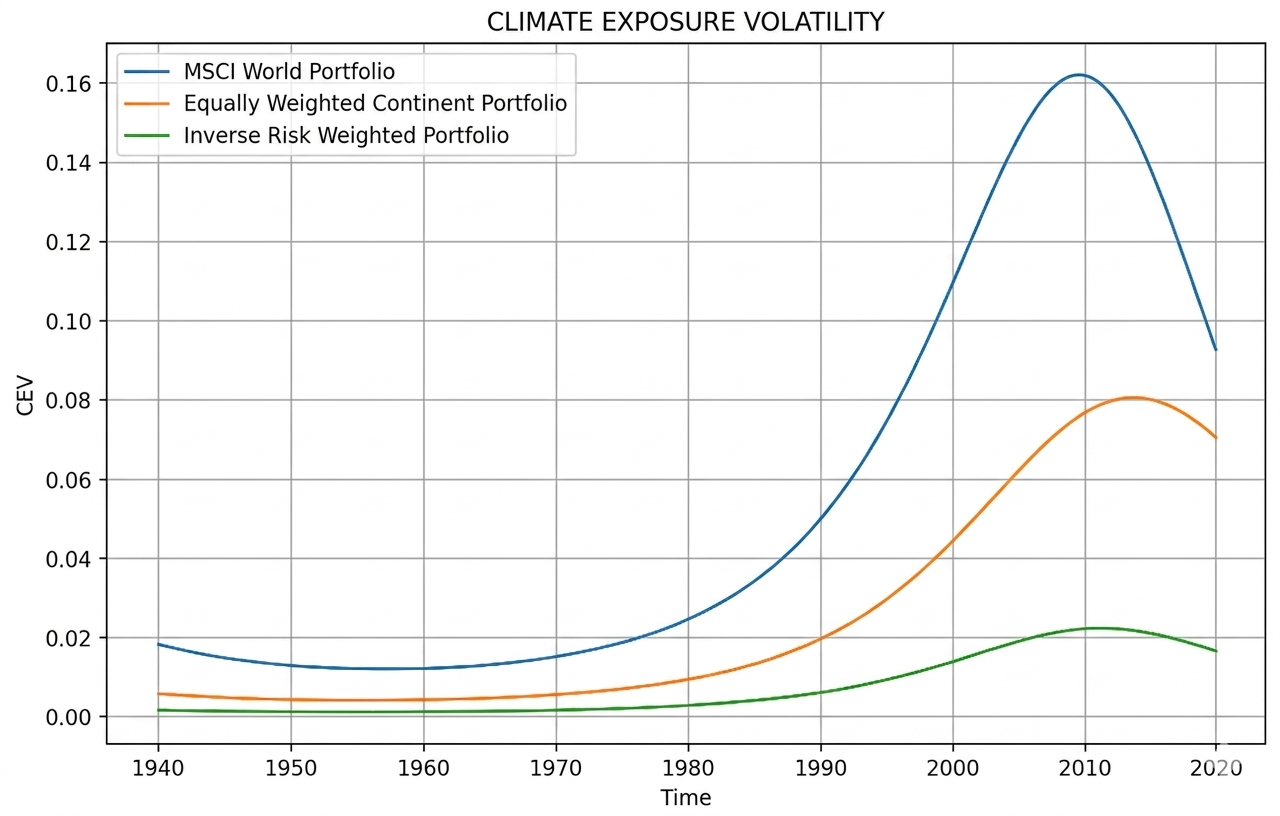} 
    \caption{Evolution of Climate Exposure Volatility across three portfolio strategies.}
    \label{fig:CEV}
\end{figure}

The provided Figure \ref{fig:CEV} illustrates the evolution of CEV over time, comparing three distinct portfolio strategies to highlight different approaches to managing physical climate risk. Specifically, the analysis considers: i) the  MSCI World portfolio; ii) the equally-weighted continent portfolio, where the climate normalized weights are distributed equally across all continents;
iii) the inverse risk-weighted portfolio, where geographic weights are assigned inversely proportional to the probability of extreme temperature events in each continent, thereby tilting the portfolio away from high-risk regions. 

As clearly shown in the chart, the CEV is a function of time, reflecting the historical fluctuations in both the frequency and severity of extreme weather events. Notably, the risk measure exhibits a significant shift starting around 2010, where the volatility begins to trend downward. This phenomenon is explained by the statistical properties of the Bernoulli variables $B_{k,t}$. The variance of a Bernoulli distribution, given by $p_{k,t}(1-p_{k,t})$, reaches its maximum at $p_{k,t}$ = 0.5. Beyond this threshold, as the probability of an extreme event continues to increase ($p > 0.5$), the variance paradoxically decreases. Since 2010, the parameters of several $B_{k,t}$ processes have exceeded $0.5$ and continued to rise, consequently the system reflects a higher degree of ``certainty'' regarding the occurrence of extreme events. This leads to a reduction in statistical volatility, even as the expected frequency (and thus the total risk) remains high. 

Furthermore, the comparison reveals a clear hierarchy in risk mitigation. The equally-weighted climate portfolio consistently highlights a lower CEV compared to the MSCI World baseline, showing that simply diversifying exposure away from market-cap concentrations can stabilize the climate risk profile. Even more effective is the Optimal Risk-Weighted Portfolio, which achieves the lowest volatility among the three strategies. By tilting the weights away from the most vulnerable regions, this approach significantly curtails the portfolio’s susceptibility to climatic shocks, pointing out that active geographic optimization is the most efficient tool for minimizing physical risk exposure.

To investigate the role and the impact of the correlations on the CEV, we compute it for the MSCI World index under three distinct scenarios. Specifically, we employ the empirical correlation estimated via the Pearson coefficient, alongside the theoretical maximum and minimum correlations derived from the Fréchet-Hoeffding bounds presented in Figure \ref{fig:correlation_boundaries}.

The results, illustrated in Figure \ref{fig:msci_cev_evolution}, show that the correlation structure plays a pivotal role in determining the overall risk level. It is evident from the findings that an increase in correlation leads to a significant rise in climate risk. In particular, higher synchronization of extreme events across geographic regions markedly amplifies the portfolio's exposure, thereby reducing the benefits of geographical diversification.

\begin{figure}[H]
    \centering
    \includegraphics[width=1.0\textwidth]{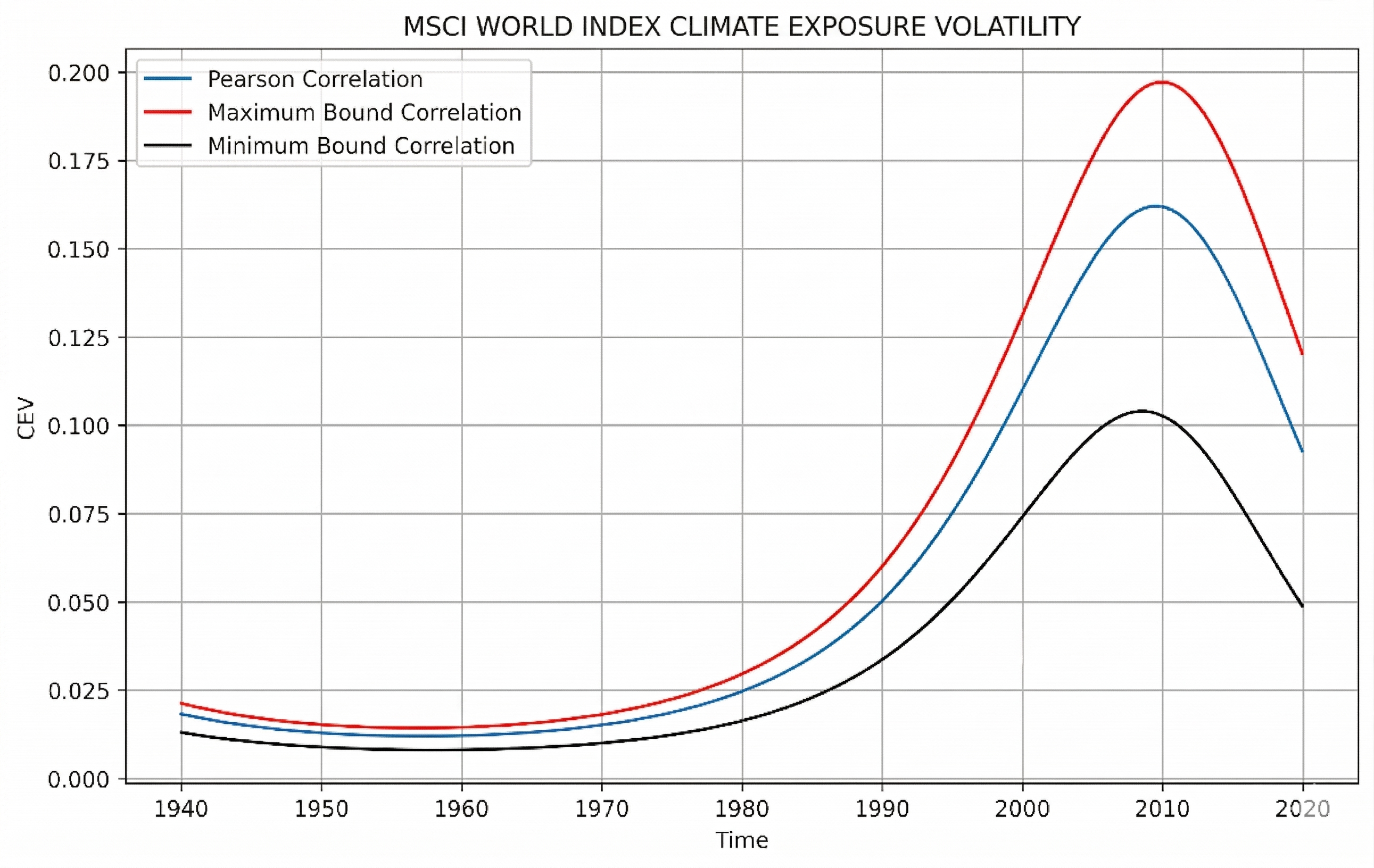} 
    \caption{Evolution of MSCI World Index Climate Exposure Volatility with three different correlation structures.}
    \label{fig:msci_cev_evolution}
\end{figure}


\section{Climate-aware portfolio optimization}
\label{sec:climate_opt}

In this section, we formalize the optimization problem designed to construct a climate-resilient portfolio. The investment universe is composed of N=120 stocks selected from the MSCI World index. This reduction of the investment universe is necessary to maintain computational tractability, balancing the representativeness of the sample with the high-dimensional complexity of the optimization algorithm. To ensure the portfolio serves as a faithful proxy for global equity markets, the selection follows a stratified sampling approach based on market capitalization and location of the headquarter. 

\subsection{The optimization problem}
The 120 assets are distributed across geographic macro-regions to reflect their relative importance in the global financial landscape. Each company is categorized based on the location of its corporate headquarter and the selection is partitioned according to weights reported in Table \ref{tab:geo_dist}.\\
\begin{table}[htbp]
\centering
\caption{Geographic Distribution of Assets in the Portfolio.}
\label{tab:geo_dist}
\begin{tabular}{lrr}
\toprule
\textbf{Region} & \textbf{Weight (\%)} & \textbf{Number of Assets} \\
\midrule
North America   & 45\%                 & 54                        \\
Europe          & 25\%                 & 30                        \\
Asia            & 25\%                 & 30                        \\
Oceania         & 5\%                  & 6                         \\
\bottomrule
\end{tabular}
\end{table}
This allocation strategy ensures that the optimization process operates on a diversified universe that mimics the capital distribution of the MSCI World. By anchoring the selection to headquarters' locations, we provide a realistic benchmark for assessing the trade-offs between financial performance and climate risk, while maintaining consistency with global indexing standards. 

We seek the portfolio $w \in \mathbb{R}^{120}$ that solves the following optimization problem:

\begin{subequations}
    \label{eq:full_optimization}
  \begin{empheq}[left=\empheqlbrace]{align}
    \min_{\bm{w}} \quad & F(\bm{w}) = [ f_1(\bm{w}), f_2(\bm{w}), f_3(\bm{w}) ]^\top \label{eq:multi_obj} \\
    \text{s.t.} \quad & \sum_{i=1}^{120} w_i = 1, \label{eq:budget_constraint} \\
    & w_i \ge 0, \quad \forall i = 1, \dots, 120 \label{eq:long_only}\;\;,
  \end{empheq}
\end{subequations}

where the three objective functions to be minimized are represented by the the expected return $f_1(\bm{w}) = -\bm{w}^\top r$, the market variance $f_2(\bm{w}) = \bm{w}^\top \Sigma \bm{w}$, and the \textit{Climate Exposure Volatility} $f_3(\bm{w})=CEV (\bm{w})$, as defined in Section \ref{sec:measures_of_risk_for_extreme_climate_events}. 

Our optimization problem considers only portfolios where the total capital is fully invested in risky assets. Consequently, the model currently excludes the possibility of allocating a portion of the wealth to a risk-free rate or engaging in short-selling. However, this framework is inherently flexible and can be generalized to include a risk-free asset and short-selling opportunities by relaxing the constraints in \eqref{eq:budget_constraint} and \eqref{eq:long_only}.

\subsection{Optimization algorithm and implementation}
\label{subsec:mopso_implementation}

To solve this multi-objective optimization problem we employ a Multi-Objective Particle Swarm Optimization (MOPSO) algorithm \citep[see][for example]{manzoni2026sustainable}. Unlike single-objective optimization, which seeks a unique global minimum, MOPSO aims to identify the Pareto Front. This front represents the set of all non-dominated solutions, that are portfolios for which it is impossible to improve one objective (e.g., increasing returns) without simultaneously worsening another (e.g., increasing risk or CEV). \\
The MOPSO is an iterative population-based algorithm where a set of candidate solutions, called particles, ``fly'' through the search space. Each particle represents a potential weight vector $\bm{w}$ and adjusts its trajectory based on its own historical best position and the best positions found by its neighbors. In a multi-objective context, the algorithm maintains an external repository, which is a specialized archive that stores the non-dominated solutions discovered during the search. The size of the population ($N_{pop}$) determines the density of the search, while the repository size ($N_{rep}$) controls the granularity and diversity of the final Pareto front. By selecting leaders from this repository to guide the swarm, the algorithm ensures that the particles converge toward the global trade-off surface while maintaining a well-distributed set of investment options. For more details about the MOPSO and the pseudo-algorithm scheme, please refer to \ref{app:mopso}.

\begin{remark}
It is worth highlighting that to ensure the matrix is well-conditioned and to prevent the smallest eigenvalues from being too close to zero, we apply the Ledoit-Wolf shrinkage method \citep{ledoit2003honey}. By employing this shrinkage estimator, we are able to reduce the condition number of the covariance matrix, calculated over the entire available dataset, to below 1,000. Maintaining a condition number of this magnitude is a widely accepted rule of thumb in financial econometrics \citep{won2013condition}.
Finally, to manage the computational complexity associated with an investment universe of N=120 assets, the CEV expressions are implemented using matrix-based operations rather than iterative loops. This vectorized approach significantly reduces the computational cost and accelerates the convergence of the MOPSO algorithm.
\end{remark}

\subsection{Algorithm parameter selection}
\label{subsec:mopso_param_selection}

To identify the optimal configuration for the MOPSO algorithm, we conduct a preliminary analysis focused on three fundamental hyperparameters, i.e., the number of iterations  $Iter \in \{200, 300, 400\}$, the population size $N_{pop}\in \{200, 500, 600\}$, and the repository size $N_{rep}\in \{100, 200, 400\}$. 

The selection of these parameters is crucial because the model is intended for a monthly rebalancing strategy. Given that the algorithm will be executed at the beginning of each month throughout a backtesting period, we must find a robust trade-off between the quality of the resulting Pareto front and the computational time required for execution.

For that task, the algorithm is run for every possible combination of the aforementioned parameters. To ensure the robustness and statistical reliability of the findings, each entry in the following tables represents the average value obtained from 10 independent runs of the optimization process. The tables of this performance analysis are categorized by the number of iterations and reported in the sections below. 

\begin{table}[H] 
\centering
\caption{Algorithm performance for 200 iterations ($Iter = 200$).}
\label{tab:algorithm_performance1}
\begin{tabular}{ccccc}
\hline
\textbf{$N_{rep}$} & \textbf{$N_{pop}$} & \textbf{Time (s)} & \textbf{Hypervolume} & \textbf{Spacing} \\ \hline
100 & 200 & 19.04 & $3.01 \times 10^{-6}$ & $8.01 \times 10^{-4}$ \\
100 & 500 & 50.36 & $4.47 \times 10^{-6}$ & $9.26 \times 10^{-4}$ \\
100 & 600 & 58.54 & $4.92 \times 10^{-6}$ & $9.39 \times 10^{-4}$ \\ \hline
200 & 200 & 30.16 & $8.85 \times 10^{-6}$ & $7.17 \times 10^{-4}$ \\
200 & 500 & 74.62 & $1.08 \times 10^{-5}$ & $7.21 \times 10^{-4}$ \\
200 & 600 & 86.25 & $1.03 \times 10^{-5}$ & $7.50 \times 10^{-4}$ \\ \hline
400 & 200 & 47.52 & $2.04 \times 10^{-5}$ & $5.54 \times 10^{-4}$ \\
400 & 500 & 115.17 & $2.29 \times 10^{-5}$ & $5.70 \times 10^{-4}$ \\
400 & 600 & 140.62 & $2.48 \times 10^{-5}$ & $6.15 \times 10^{-4}$ \\ \hline
\end{tabular}
\end{table}

\begin{table}[H]
\centering
\caption{Algorithm performance for 300 iterations ($Iter = 300$).}
\label{tab:algorithm_performance2}
\begin{tabular}{ccccc}
\hline
\textbf{$N_{rep}$} & \textbf{$N_{pop}$} & \textbf{Time (s)} & \textbf{Hypervolume} & \textbf{Spacing} \\ \hline
100 & 200 & 27.64 & $4.10 \times 10^{-6}$ & $9.06 \times 10^{-4}$ \\
100 & 500 & 71.91 & $4.99 \times 10^{-6}$ & $9.16 \times 10^{-4}$ \\
100 & 600 & 115.21 & $5.41 \times 10^{-6}$ & $8.82 \times 10^{-4}$ \\ \hline
200 & 200 & 42.96 & $9.48 \times 10^{-6}$ & $7.50 \times 10^{-4}$ \\
200 & 500 & 107.04 & $1.11 \times 10^{-5}$ & $7.65 \times 10^{-4}$ \\
200 & 600 & 160.93 & $1.10 \times 10^{-5}$ & $7.35 \times 10^{-4}$ \\ \hline
400 & 200 & 73.74 & $1.91 \times 10^{-5}$ & $5.57 \times 10^{-4}$ \\
400 & 500 & 217.39 & $3.37 \times 10^{-5}$ & $6.16 \times 10^{-4}$ \\
400 & 600 & 257.48 & $3.17 \times 10^{-5}$ & $6.10 \times 10^{-4}$ \\ \hline
\end{tabular}
\end{table}

\begin{table}[H]
\centering
\caption{Algorithm performance for 400 iterations ($Iter = 400$).}
\label{tab:algorithm_performance3}
\begin{tabular}{ccccc}
\hline
\textbf{$N_{rep}$} & \textbf{$N_{pop}$} & \textbf{Time (s)} & \textbf{Hypervolume} & \textbf{Spacing} \\ \hline
100 & 200 & 42.27 & $3.68 \times 10^{-6}$ & $8.97 \times 10^{-4}$ \\
100 & 500 & 95.20 & $4.38 \times 10^{-6}$ & $9.53 \times 10^{-4}$ \\
100 & 600 & 123.88 & $5.03 \times 10^{-6}$ & $8.80 \times 10^{-4}$ \\ \hline
200 & 200 & 58.95 & $9.39 \times 10^{-6}$ & $7.21 \times 10^{-4}$ \\
200 & 500 & 149.82 & $1.15 \times 10^{-5}$ & $7.26 \times 10^{-4}$ \\
200 & 600 & 186.80 & $1.18 \times 10^{-5}$ & $7.43 \times 10^{-4}$ \\ \hline
400 & 200 & 100.14 & $2.40 \times 10^{-5}$ & $5.99 \times 10^{-4}$ \\
400 & 500 & 243.93 & $4.03 \times 10^{-5}$ & $6.30 \times 10^{-4}$ \\
400 & 600 & 287.35 & $4.70 \times 10^{-5}$ & $6.61 \times 10^{-4}$ \\ \hline
\end{tabular}
\vspace{5pt}
\end{table}

Based on the results reported in Tables \ref{tab:algorithm_performance1}, \ref{tab:algorithm_performance2} and \ref{tab:algorithm_performance3}, we can derive several key insights regarding the impact of hyperparameters on the optimization process. In such tables, we consider the hypervolume metric \citep{guerreiro2021hypervolume}, which is a measure of both convergence and diversity with a higher hypervolume describing a superior distribution of solutions across the objective space, and the spacing indicator that quantifies the uniformity of the solutions' distribution along the Pareto frontier by measuring the standard deviation of the distances between adjacent non-dominated vectors \citep{schott1995fault}. A spacing value approaching zero indicates that the solutions are evenly spread throughout the objective space.

Increasing the repository size ($N_{rep}$) from 100 to 200 and 400 leads to a major improvement in the optimization quality. Specifically, by moving from 100 to 200, the hypervolume increases by an order of magnitude, and the improvement remains substantial when further expanding the repository to 400. For this reason, we choose the configuration with $N_{rep}$ = 400, which yields clearly superior results while maintaining a computationally acceptable execution time.
    
Looking at the  population size  we see that a substantial gain in hypervolume is observed when moving from 200 to 500 particles, whereas the improvement becomes less significant when further increasing the size to 600. Hence, we choose $N_{pop}$= 500.
    
 Moreover, the number of iterations significantly influences the quality of the solution. Although the spacing metric remains relatively stable and independent of this parameter, the hypervolume shows a marked increase from $Iter = 200$ to $Iter = 300$, and a further, though less pronounced, improvement from $Iter = 300$ to $Iter = 400$. As expected, a higher number of iterations entails an increase in total execution time. 

After some preliminary tests, the configuration with $Iter = 300$, $N_{rep}$ = 400, and $N_{pop}$ = 500 is the one we deem most appropriate for our analysis. While performance is consistent across the alternatives, this setup yields the lowest spacing value and is characterized by a lower computational cost, thus representing what we consider the optimal trade-off between efficiency and solution quality.

In Figure \ref{fig:Pareto_Frontier_2020_fp}, we present the resulting three-dimensional Pareto frontier computed at January 1, 2020. The visualization maps the complex interactions between the three conflicting objectives: the x-axis represents the portfolio variance, the y-axis shows the \textit{Climate Exposure Volatility}, and the z-axis denotes the expected returns. 

The quality of the frontier, characterized by a highly uniform distribution of points, its smooth and concave shape, and its broad extension across the objective space, confirms the effectiveness of the chosen optimization algorithm and the appropriate calibration of its parameters. 
Furthermore, the visualization clearly illustrates that the classic mean-variance trade-off remains fundamental, even with the integration of CEV. It is evident that to achieve a higher target for expected returns, one must necessarily accept an increase in portfolio risk, specifically in its variance. 

\begin{figure}[H]
    \centering
    \includegraphics[width=0.6\textwidth]{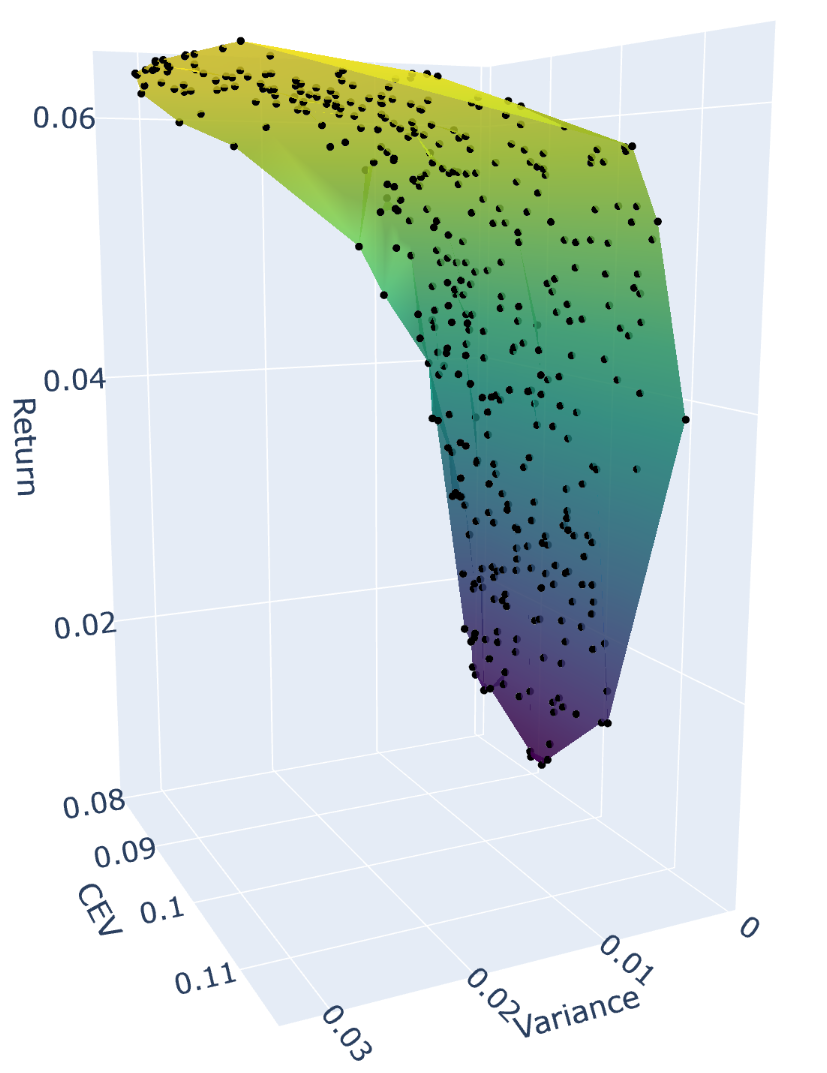} 
    \caption{Pareto frontier for the 3-objective optimization problem. For a detailed visualization of the surface from multiple perspectives, please refer to \ref{app:different_perspectives}.}
    \label{fig:Pareto_Frontier_2020_fp}
\end{figure}

Beyond the objective space, we investigate the composition of the portfolios belonging to the frontier. Specifically, out of a total universe of 120 stocks, we observe that the average number of active positions is 22. This calculation is based on a threshold of $10^{-3}$, where weights below this value are considered negligible and treated as zero. A deeper analysis reveals that for portfolios with low variance, this number tends to be higher, reaching approximately 55 active weights. However, as we move along the frontier toward higher variance levels, the number of non-zero weights decreases. This trend further confirms that a reduction in diversification leads to an increase in risk and potentially to higher returns. 

In contrast, when considering a standard mean-variance frontier that neglects CEV, the average number of non-zero weights drops to 14.\footnote{Results are available upon request.} This significant difference suggests that incorporating an additional risk dimension into the optimization framework emphasizes diversification as a primary tool for risk mitigation, preventing the optimizer from over-concentrating in assets that may appear financially optimal but remain highly vulnerable to climate shocks. 

\begin{remark}
We emphasize that integrating the CEV into the standard mean-variance framework does not overparameterize the problem by introducing something already captured by financial variance. Physical risk events are not assumed to be stationary, as is typically the case for variance in most financial and econometric models (and in the Mean-Variance framework we considered). The distinction is necessary to capture the time-varying dynamics of climate change damages, which are worsening at a rate faster than linear. This can already be corroborated by the data: as discussed above, incorporating the CEV substantially improves portfolio diversification, a finding further corroborated by the backtesting analysis in Section \ref{sec:backtesting}, where we will show that introducing this additional objective does not penalize portfolio performance in mean-variance terms and in some cases even yields improvements.
\end{remark}
To quantify the structural differences between the allocations, we compute the continuous Jaccard \citep[see e.g.,][]{kamiura2023jaccard} distance between the various portfolios on the frontier using the following formulation: 

\begin{equation}
d_J(x, y) = 1 - \frac{\sum_{i} \min(x_i, y_i)}{\sum_{i} \max(x_i, y_i)}.
\end{equation}

The results for the two problems are presented in the heatmap in Figure \ref{fig:comparison_heatmaps}, where portfolios are sorted by increasing order of variance. We report the Mean-Variance-CEV Jaccard distance heatmap on the left and the Mean-Variance on the right. As expected, portfolios with significantly different variances tend to have highly distinct allocations, as evidenced by the high Jaccard distance values. Interestingly, the interference or influence of the third objective (CEV) is clearly visible, as it creates local variations in allocation patterns that deviate from the standard Mean-Variance case.

\begin{figure}[H]
    \centering
    \begin{subfigure}[b]{0.48\textwidth}
        \centering
        \includegraphics[width=\textwidth]{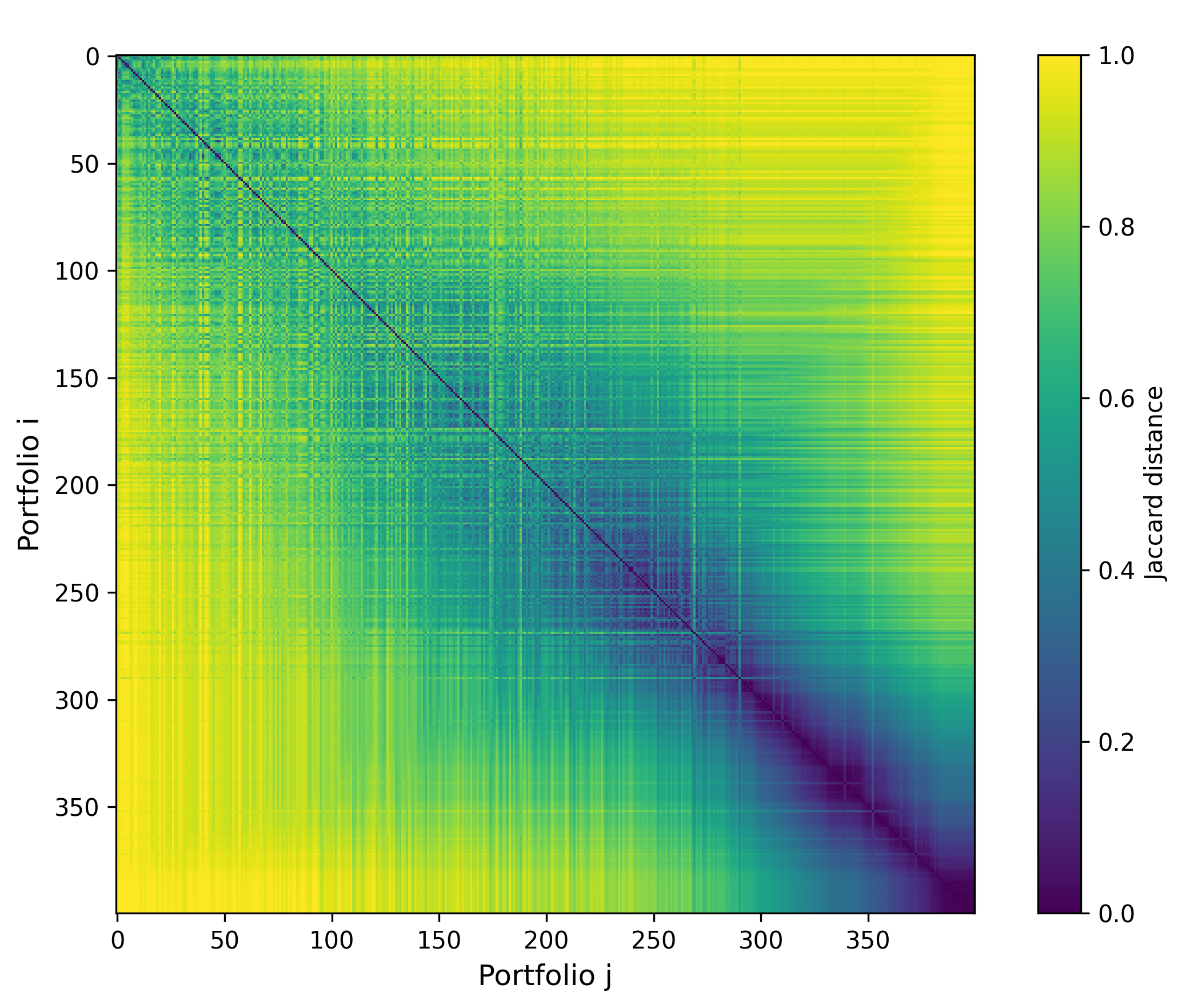}
        \caption{Mean-Variance-CEV.}
        \label{fig:jaccard_heatmap}
    \end{subfigure}
    \hfill 
    \begin{subfigure}[b]{0.48\textwidth}
        \centering
        \includegraphics[width=\textwidth]{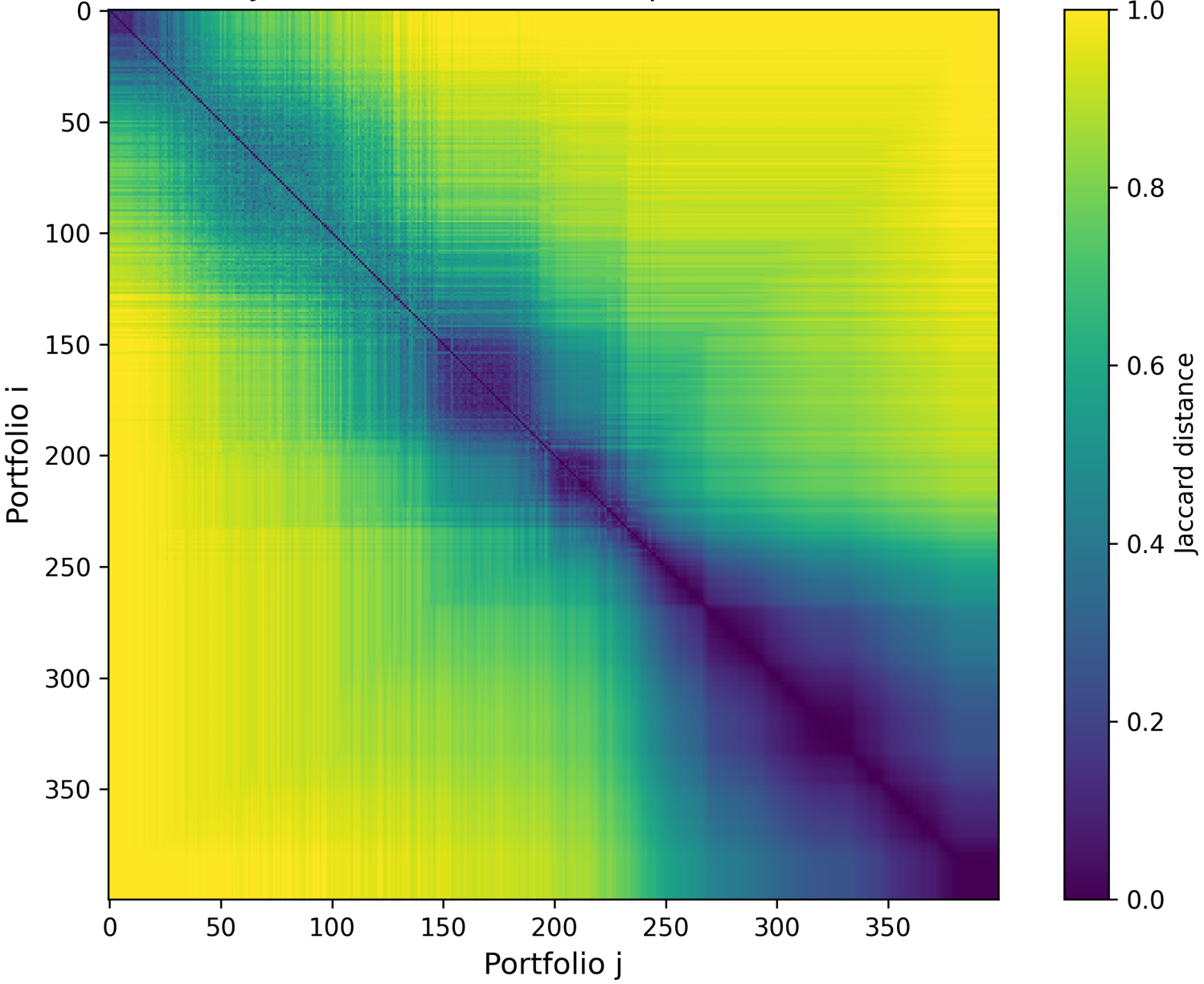}
        \caption{Mean-Variance.}
        \label{fig:heat_mv}
    \end{subfigure}
    
    \caption{Comparison of Jaccard distance heatmaps.}
    \label{fig:comparison_heatmaps}
\end{figure}

In order to further investigate the relationship between \textit{Climate Exposure Volatility} and the portfolio's expected return, we perform a sensitivity analysis by slicing the 3D Pareto frontier into distinct variance intervals. Specifically, the total range of traditional financial variance is divided into several sub-intervals. For each interval, we select the particles (portfolios) falling within that specific range and project them onto a 2D plot. \\
In these visualizations, the x-axis represents the CEV, while the y-axis represents the expected return. To provide a clearer trend analysis, we also plot an interpolated parabolic curve. 

The results, as illustrated in Figure \ref{fig:Cev_vs_Return}, reveal a clear and consistent pattern: an increase in expected return is systematically associated with an increase in CEV. This confirms that CEV operates as an additional dimension of risk. Just as traditional variance measures standard market volatility, CEV quantifies a portfolio's vulnerability to climate-related shocks. Crucially, however, the reduction in expected return associated with minimizing CEV appears limited (see also Figure \ref{fig:Pareto_Frontier_2020_fp}). This supports the premise that investors can mitigate the variance induced by physical climate risks without significantly compromising overall portfolio performance.


\begin{figure}[H]
    \centering
    \includegraphics[width=0.9\textwidth]{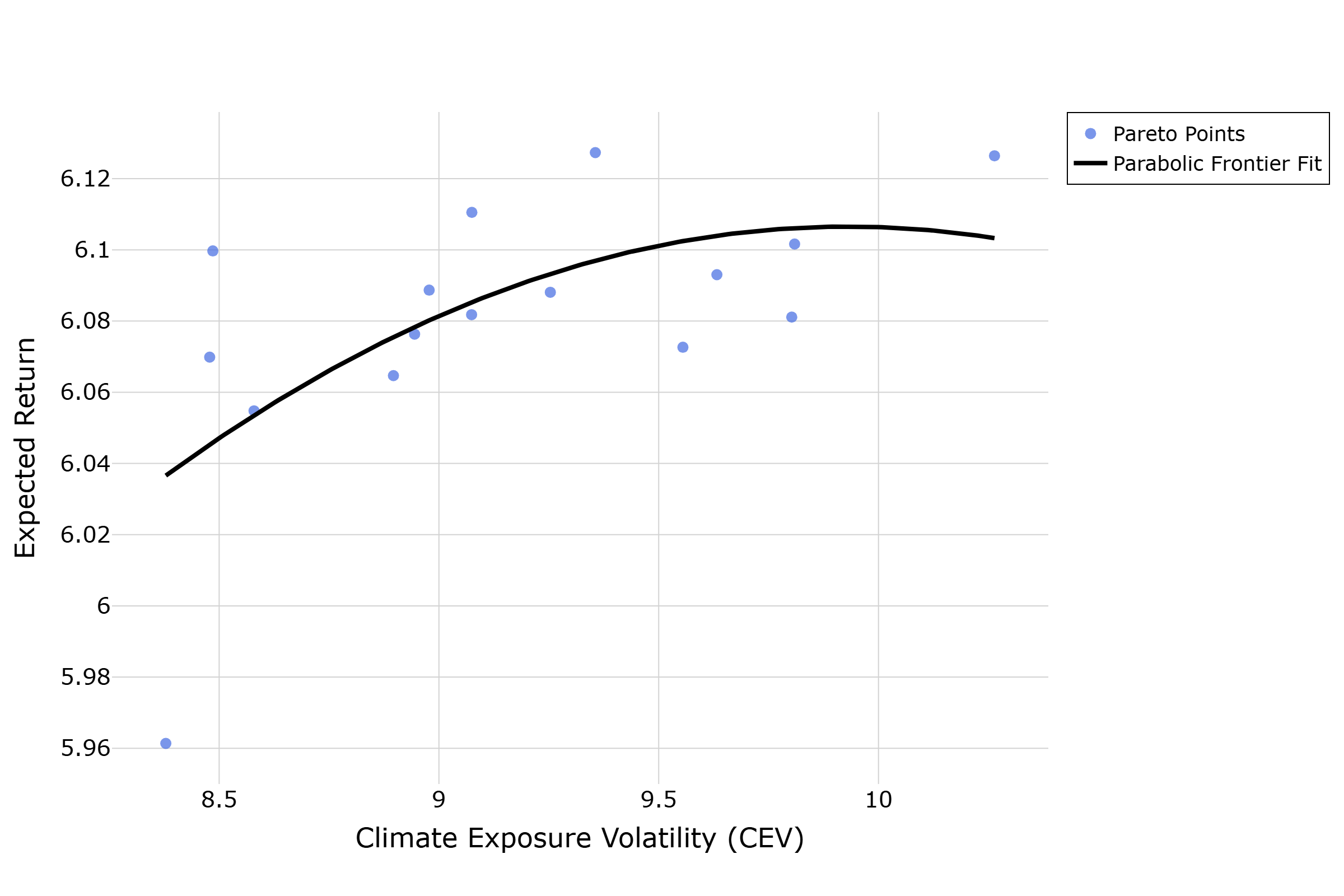} 
    \caption{Expected return vs. CEV: 2D cross-sectional Pareto projections (in percentage).}
    \label{fig:Cev_vs_Return}
\end{figure}

Similar considerations could be stated when replacing the CEV with the \textit{Climate Risk Exposure} (CRE). We refer the interested reader to \ref{app:cre_pareto_front}.


\section{Backtesting analysis}
\label{sec:backtesting}

In this section, we evaluate the ex-post performances of the portfolios from the Pareto frontier derived in Section \ref{sec:climate_opt}, which explicitly incorporates climate risk. Specifically, the backtesting analysis covers the period from January 1, 2020, to April 30, 2025 considering an in-sample rolling-window starting from January 1, 2015. We execute the analysis on an Intel Core i5 with 16 GB RAM and Windows 10 Pro system using PyCharm 2025.2.3. We also recall that the time execution is 217.39 seconds per month.

On the first business day of each month within the backtesting period, the Pareto frontier is re-optimized running the MOPSO algorithm described in the previous section: we simultaneously minimize the market variance and the CEV, while maximizing the expected return. The input parameters are updated monthly using a 5-year rolling-window approach. The expected returns are calculated as the mean of the monthly returns over the preceding five years, and the covariance matrix is estimated using daily returns from the same five-year look-back period. Furthermore, the temperature anomaly probabilities for each continent are updated using the most recent available data, specifically those recorded in the month immediately preceding the rebalancing date. 

Following this monthly re-optimization on the first business day, the performance of the selected portfolios is evaluated at the end of the month. This procedure is repeated iteratively: at the start of each subsequent month, the portfolios are recalculated and rebalanced based on the newly updated Pareto frontier. 

In our analysis, we include market-cap-weighted portfolio and the equally-weighted portfolio. Furthermore, we examine specific portfolios on the Pareto frontier, such as the minimum market variance portfolio, minimum \textit{Climate Exposure Volatility} portfolio and the maximum expected return portfolio. 

To identify additional portfolios on the frontier that balance all three competing objectives, i.e. expected return, variance, and CEV, we apply a min-max normalization to each objective. This ensures that all components are scaled within the $[0, 1]$ range while preserving their relative ordering. We then combine these normalized objectives into a single scalar function using a convex combination governed by three positive parameters $(a, b, c)$. Since our optimization framework is structured as a minimization problem, we select the portfolio on the Pareto Frontier that yields the minimum value of this combined objective function. This approach allows us to explore different strategic tilts by varying the weights assigned to each objective. Specifically, we test the  weight combinations reported in Table \ref{tab:weight_scenarios}.

\begin{table}[htbp]
\centering
\caption{Weighting portfolios for multi-objective optimization ($a, b, c$).}
\label{tab:weight_scenarios}
\begin{tabular}{lcccl}
\toprule
\textbf{Portfolio} & \bm{$a$} & \bm{$b$} & \bm{$c$} & \textbf{Investment Strategy} \\
\midrule
Balanced      & $1/3$ & $1/3$ & $1/3$ & Equal importance to all objectives \\
Minimum Risk  & $0$   & $1/2$ & $1/2$ & Financial and climate risk awareness \\
Return-oriented & $1/2$ & $1/4$ & $1/4$ & Priority to expected returns \\
Variance-oriented & $1/4$ & $1/2$ & $1/4$ & Priority to variance minimization\\
Climate-oriented  & $1/4$ & $1/4$ & $1/2$ & Priority to climate risk awareness\\
Mean-Variance & $1/2$ & $1/2$ & $0$   & Traditional Mean-Variance \\
\bottomrule
\end{tabular}
\end{table}

\vspace{0.4cm}

Figure \ref{fig:backtest} illustrates the cumulative returns of the various portfolio strategies over the analyzed period.
\begin{figure}[H]
    \centering
    \makebox[\textwidth][c]{
        \includegraphics[width=1.1\textwidth, keepaspectratio]{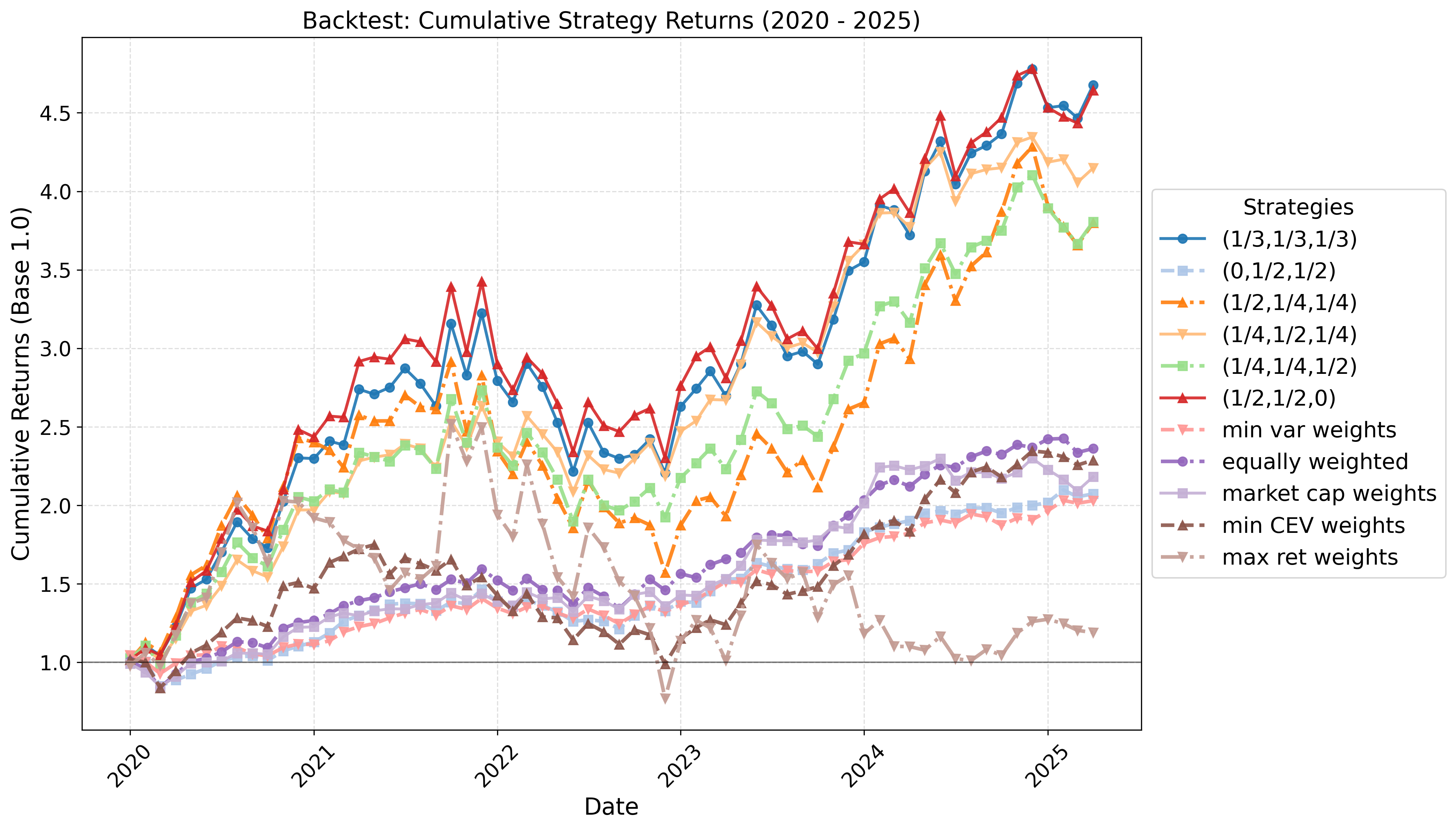}
    }
    \caption{Cumulative returns of the Pareto-optimized portfolios and benchmarks (January 2020 -- April 2025).}
    \label{fig:backtest}
\end{figure}
We observe that minimum variance, equally-weighted, and market-cap-weighted portfolios exhibit comparable performance, characterized by a steady upward trajectory and remarkably low volatility. In contrast, portfolios incorporating climate risk and expected returns highlight a trade-off profile, achieving superior returns at the expense of higher variance

To better evaluate the risk-adjusted performance of these portfolios, we measure several key financial metrics: the Compound Annual Growth Rate (CAGR), the Sharpe Ratio, the Sortino Ratio, the Maximum Drawdown and the realized CRE and CEV (that we call r-CRE and r-CEV, respectively) over time. In our analysis, we assume a risk-free rate of zero for the calculation of both the Sharpe and Sortino ratios. The Maximum Drawdown represents the largest percentage drop from a peak to a trough, indicating the worst-case loss experienced during the period (see \citet{sortino2001managing} and \citet{aprea2025neural}). Finally, the r-CRE and r-CEV are derived by multiplying the climate normalized weights by the actually realized temperature anomalies, these products are then summed up across all continents to obtain a monthly aggregate, the r-CRE  is then computed as the sample mean of these aggregates over time, while the  r-CEV  as the sample variance. These metrics are summarized in  Table \ref{tab:performance_metrics}.

\begin{table}[H]
    \centering
    \caption{Performance Metrics Comparison (January 2020 -- April 2025)}
    \label{tab:performance_metrics}
    \begin{tabular}{lcccccc}
        \toprule
        \textbf{Strategy} & \textbf{CAGR} & \textbf{Sharpe } & \textbf{Sortino } & \textbf{ Drawdown} & \textbf{r-CRE} & \textbf{r-CEV} \\
        \midrule
        Market-Cap Weights   & 15.75\% & 1.041 & 1.850 & 14.78\% & 0.714 & 10.64\% \\
        Equally-Weighted     & 17.49\% & 1.153 & 1.834 & 16.72\% & 0.724 & 10.96\% \\
        Min Market Variance  & 14.20\% & 1.282 & 2.324 & \textbf{11.47\%} & 0.762 & 13.56\% \\
        Min CEV & 16.68\% & 0.739 & 1.207 & 43.48\% & 
        0.710 & 10.50\% \\
        Max Expected Return  & 3.32\%  & 0.320 & 0.674 & 69.37\% & \textbf{0.695} & 11.87\% \\
        (1/3, 1/3, 1/3)      & \textbf{33.54\%} & 1.153 & 2.787 & 31.54\% & 0.721 & \textbf{10.48\%} \\
        (0, 1/2, 1/2)        & 14.65\% & 1.133 & 1.304 & 17.38\% & 0.716 & 11.18\% \\
        (1/2, 1/4, 1/4)      & 28.43\% & 0.920 & 2.077 & 46.10\% & 0.727 & 10.51\% \\
        (1/4, 1/2, 1/4)      & 30.57\% & \textbf{1.312} & \textbf{3.029} & 20.78\% & 0.723 & 10.78\% \\
        (1/4, 1/4, 1/2)      & 28.47\% & 1.045 & 2.382 & 30.57\% & 0.716 & 10.70\% \\
        (1/2, 1/2, 0)        & 33.35\% & 1.136 & 2.356 & 32.81\% & 0.698 & 10.77\% \\
        \bottomrule
    \end{tabular}
\end{table}
The analysis of the performance metrics provides several key insights into the behavior of the different strategies. The maximum expected return portfolio exhibits the worst performance across all evaluated metrics (except for the climate risk metrics). This confirms that selecting a portfolio solely based on expected returns, without incorporating any risk measure, leads to extreme volatility and risk exposure. Indeed, it records the worst maximum drawdown of 69.37\%, and a Sharpe Ratio of only 0.320, which is significantly lower than all other strategies.

The equally-weighted and market-cap weighted portfolios, used as benchmarks, show similar trends, reflecting the general market performance over the last 5 years. The minimum CEV portfolio exhibits a relatively similar return. However, while it effectively minimizes climate risk resulting in one of the lowest realized CEV and CRE despite the latter not being explicitly targeted in the minimization, it struggles with volatility control showing the third worst Sharpe Ratio and a high Maximum Drawdown. In contrast, the minimum market variance portfolio, while achieving a lower CAGR compared to the benchmarks, turns out to be the most effective in controlling market volatility. It records a Maximum Drawdown of 11.47\%, the best among all observed portfolios.

Similarly, the $(0, 1/2, 1/2)$ strategy, which ignores expected returns to assign equal weight to financial and climate variance, displays a profile very close to the Minimum Variance portfolio, achieving the second best Maximum Drawdown at 17.38\% and the lowest average CEV.

When considering portfolios on the frontier that include expected return in their objective function, we observe the highest CAGR values. The $(1/3, 1/3, 1/3)$ strategy, which weights all three objectives equally, achieves the highest overall CAGR (33.54\%). The $(1/2, 1/2, 0)$ portfolio shows very similar results, although it does not explicitly weight the CEV in the convex combination (achieving indeed worst CEV than the  $(1/3, 1/3, 1/3)$ strategy), it still benefits from being selected from a Pareto frontier, originally constructed including this climate risk metric. 

While these CAGR values (around 30\%) may appear exceptionally high, it is important to note that this backtest does not account for transaction costs, bid-ask spreads, and in general rebalancing costs, which would realistically reduce net performance.

Finally, the remaining three portfolios show distinct behaviors. The $(1/2, 1/4, 1/4)$ strategy records the second-worst Maximum Drawdown (46.10\%), suggesting that its weights may over-prioritize returns at the expense of risk mitigation. Conversely, the $(1/4, 1/2, 1/4)$ portfolio emerges as a reasonable balanced choice, as it yields good results accross all six metrics, suggesting that this specific weight distribution provides one of the most effective compromise between the three competing objectives.

In conclusion, the backtesting results underscore the fundamental trade-offs inherent in multi-objective portfolio optimization. While strategies focused exclusively on expected returns prove to be unsustainable due to extreme drawdowns, the integration of climate risk metrics does not necessarily come at the expense of financial performance. On the contrary, portfolios that balance market risk, climate exposure, and expected returns, most notably the $(1/4, 1/2, 1/4)$ and $(1/3, 1/3, 1/3)$ configurations outline superior resilience and efficiency. 

These findings indicate that incorporating climate risk within the optimization framework serves as an effective diversification mechanism and does not necessarily deteriorate portfolio returns. Specifically, accounting for climate-related risk factors appears to enhance the diversification properties of the portfolio by introducing an additional dimension of risk consideration beyond traditional financial metrics. Despite this additional constraint, the empirical results do not suggest a consequent drop in return performance. The integration of climate-risk considerations can be achieved while maintaining competitive return levels, highlighting the potential compatibility between climate-aware investment strategies and traditional portfolio performance goals.

\vspace{0.7cm}

\section{Conclusions} \label{sec:conclusion}
\color{black}
This study has provided a flexible and robust quantitative foundation for climate-resilient portfolio construction, offering a viable tool for asset managers to navigate the escalating challenges of physical climate risk.

The work has investigated the impact of physical climate risk on global equity portfolios. Our empirical analysis has confirmed that extreme temperature anomalies represent a material financial threat, exerting a statistically significant negative pressure on asset returns across most economic sectors. 
This financial impact is deeply rooted in the non-stationary and intensifying nature of environmental hazards (see Figure \ref{fig:temperature_anomalies}). Indeed, our results have confirmed that climate risk is not a static outlier but a dynamic variable that has significantly escalated in recent decades. Consequently, these findings have invalidated the traditional assumption of a constant risk distribution, highlighting the urgent need for a dynamic approach to portfolio risk management that accounts for the evolving frequency and intensity of climate shocks.

The core contribution of this research lies in the introduction of two dynamic metrics: \textit{Climate Risk Exposure}  and \textit{Climate Exposure Volatility} (see Section \ref{sec:measures_of_risk_for_extreme_climate_events}). Unlike existing static country-level indices, these measures capture the time-varying probability of extreme events and the systemic risk arising from their spatial correlation across continents. Furthermore, by incorporating the asset intensity of each firm, these metrics account for firm-specific vulnerability, reflecting the varying degrees of dependency on physical assets. This effectively distinguishes between capital-heavy sectors like real estate and more agile industries like retailers, while considering the precise geographical distribution leveraging corporate revenues.

By integrating these measures into a Multi-Objective Particle Swarm Optimization  framework, we have highlighted that investors can explicitly account for climate resilience within the mean-variance paradigm, without necessarily sacrificing competitive expected return targets. We have investigated the non-linear trade-offs between expected returns, market variance, and climate risk, discussing how active geographic optimization is an efficient tool for minimizing vulnerability to physical shocks. Finally, our backtesting analysis has confirmed the practical benefits of this approach, showing that climate-aware strategies can enhance portfolio resilience in a dynamic market environment that include the COVID-19 pandemic and the 2022 energy shock, generating competitive performance compared to traditional benchmarks. In particular, the results suggest that integrating the \textit{Climate Exposure Volatility} does not necessarily impair financial returns, rather it provides an additional layer of risk management that helps identify more efficient portfolios along the Pareto frontier, offering ways to reduce both overall volatility and maximum drawdown while maintaining a competitive targets of expected return and diversification.

\bibliographystyle{elsarticle-harv}
\bibliography{bibliography} 

\appendix

\section{Proof of correlation bounds for Bernoulli variables} 
\label{app:corr_bound_proof}

In this section of the appendix we provide the full proof of the corollary regarding the correlation bounds for two Bernoulli random variables.

\begin{corollary}
\label{cor:bernoulli_bounds_app}
The correlation coefficient $\rho$ between two Bernoulli random variables $X_1 \sim \text{Bern}(p_1)$ and $X_2 \sim \text{Bern}(p_2)$ is bounded by $\rho_{min} \le \rho \le \rho_{max}$, where
\begin{equation}
    \rho_{min} = \frac{\max(0, p_1 + p_2 - 1) - p_1 p_2}{\sqrt{p_1(1-p_1)} \sqrt{p_2(1-p_2)}}
\end{equation}
\begin{equation}
    \rho_{max} = \frac{\min(p_1, p_2) - p_1 p_2}{\sqrt{p_1(1-p_1)} \sqrt{p_2(1-p_2)}} \;\;.
\end{equation}
\end{corollary}

\textit{Proof}. To derive the explicit forms of $\rho_{min}$ and $\rho_{max}$, we start from the Fréchet–Hoeffding bounds \\
\begin{equation}
    \text{corr}(F_1^{-1}(U), F_2^{-1}(1-U)) \leq \text{corr}(X_1, X_2) \leq \text{corr}(F_1^{-1}(U), F_2^{-1}(U)) \, ,
\end{equation} \\
with $U \sim \text{Unif}([0, 1])$, and we apply it to two Bernoulli random variable $X_1 \sim \text{Bern}(p_1)$ and $X_2 \sim \text{Bern}(p_2)$ with respective cumulative distribution function $F_1^{-1}$ and $F_2^{-1}$. 

For a Bernoulli variable $X \sim \text{Bern}(p)$, the quantile function is
\begin{equation}
    F^{-1}(u) = 
    \begin{cases} 
    0 & \text{if } 0 \le u < 1-p \\
    1 & \text{if } 1-p \le u \le 1\, .
    \end{cases}
\end{equation}

First of all, to estimate the upper bound $\rho_{max}$, we consider 
the product $F_1^{-1}(u) F_2^{-1}(u)$ which is equal to 1 if and only if
\begin{equation}
    u \ge 1-p_1 \quad \text{and} \quad u \ge 1-p_2 \implies u \ge \max(1-p_1, 1-p_2) = 1 - \min(p_1, p_2)
\end{equation}
and integrating over $[0, 1]$ one has
\begin{equation}
    E[F_1^{-1}(U) F_2^{-1}(U)] = \int_{1-\min(p_1, p_2)}^{1} 1 \, du = \min(p_1, p_2).
\end{equation}
Thus, the $\text{Cov}(F_1^{-1}(U), F_2^{-1}(U))$ is 

\begin{equation}
\begin{aligned}
    \text{Cov}(F_1^{-1}(U), F_2^{-1}(U)) &= E[F_1^{-1}(U) F_2^{-1}(U)] - E[F_1^{-1}(U)] \cdot E[F_2^{-1}(U)] \\
    &= E[F_1^{-1}(U) F_2^{-1}(U)] - E[X_1]E[X_2] \\
    &= \min(p_1, p_2) - p_1 p_2.
\end{aligned}
\end{equation}

Dividing by the product of standard deviations $ \sqrt{p_1(1-p_1)} \sqrt{p_2(1-p_2)}$, we obtain $\text{corr}(F_1^{-1}(U), F_2^{-1}(U))$, leading to 
\begin{equation}
    \rho_{max} = \frac{\min(p_1, p_2) - p_1 p_2}{\sqrt{p_1(1-p_1)} \sqrt{p_2(1-p_2)}}.
\end{equation}

Now, to calculate the lower bound $\rho_{min}$, we examine 
the product $F_1^{-1}(u) F_2^{-1}(1-u)$ that is 1 if and only if
\begin{equation}
    u \ge 1-p_1 \quad \text{and} \quad 1-u \ge 1-p_2 \implies 1-p_1 \le u \le p_2\, .
\end{equation}
This interval has a non-zero length only if $1-p_1 \le p_2$, which is equivalent to $p_1 + p_2 - 1 \ge 0$. The length of the interval is $\max(0, p_2 - (1-p_1)) = \max(0, p_1 + p_2 - 1)$. 

Integrating over [0,1] one gets
\begin{equation}
    E[F_1^{-1}(U) F_2^{-1}(1-U)] = \int_{1-p_1}^{p_2} 1 \, du = \max(0, p_1 + p_2 - 1)
\end{equation}
and thus, the $\text{Cov}(F_1^{-1}(U), F_2^{-1}(\-U))$ is 

\begin{equation}
\begin{aligned}
    \text{Cov}(F_1^{-1}(U), F_2^{-1}(1-U)) &= E[F_1^{-1}(U) F_2^{-1}(1-U)] - E[F_1^{-1}(U)] \cdot E[F_2^{-1}(1-U)] \\
    &= E[F_1^{-1}(U) F_2^{-1}(1-U)] - E[X_1]E[X_2] \\
    &= \max(0, p_1 + p_2 - 1) - p_1 p_2\, .
\end{aligned}
\end{equation}

Dividing by the product of standard deviations $ \sqrt{p_1(1-p_1)} \sqrt{p_2(1-p_2)}$, we obtain $\text{corr}(F_1^{-1}(U), F_2^{-1}(1-U))$, leading to 
\begin{equation}
    \rho_{min} = \frac{\max(0, p_1 + p_2 - 1) - p_1 p_2}{\sqrt{p_1(1-p_1)} \sqrt{p_2(1-p_2)}}\, .
\end{equation}


\section{Tri-Objective Optimization: Integrating Climate Risk Exposure}
\label{app:cre_pareto_front}

In this section of the appendix we present the results of the multi-objective optimization performed at January 1, 2020, where we utilize the \textit{Climate Risk Exposure} (CRE) instead of the \textit{Climate Exposure Volatility} (CEV) to evaluate the portfolio's climate-related risk. By adopting CRE as our objective, we aim to measure the absolute magnitude of the portfolio's vulnerability to extreme temperature anomalies, providing an alternative perspective on environmental risk.

The optimization framework simultaneously considers three competing objectives: expected return, market variance, and CRE. In line with the methodology described in Section \ref{sec:climate_opt}, expected returns and the covariance matrix are estimated using a 5-year rolling window of historical data. Furthermore, the probabilities of extreme temperature anomalies for each continent used in the CRE are sourced from the latest available data prior to the optimization date.

 Figure \ref{fig:persp_1_cre} displays the 3D Pareto frontier. This surface, similarly to the one in Figure \ref{fig:Pareto_Frontier_2020_fp}, illustrates the fundamental trade-offs in a multidimensional space, where any improvement in climate resilience or financial variance necessitates a compromise in expected return.

To analyze the structural composition of these optimal solutions, Figure \ref{fig:persp_2_cre} presents the Jaccard distance heatmap. The heatmap clearly highlights that portfolios with different objective values possess distinct asset compositions.

Finally, in Figure \ref{fig:Cre_vs_Return}, we can observe that an increase in the expected return is associated with an increase in the CRE. This confirms that CRE acts as an additional dimension of risk, analogous to standard variance. More specifically, CRE quantifies a portfolio’s vulnerability to environmental shocks, showing that investors seeking higher performance must accept not only greater market volatility but also heightened climate risk. 

\begin{figure}[H]
    \centering
    \begin{subfigure}[b]{0.49\textwidth}
        \centering
        \includegraphics[width=\textwidth]{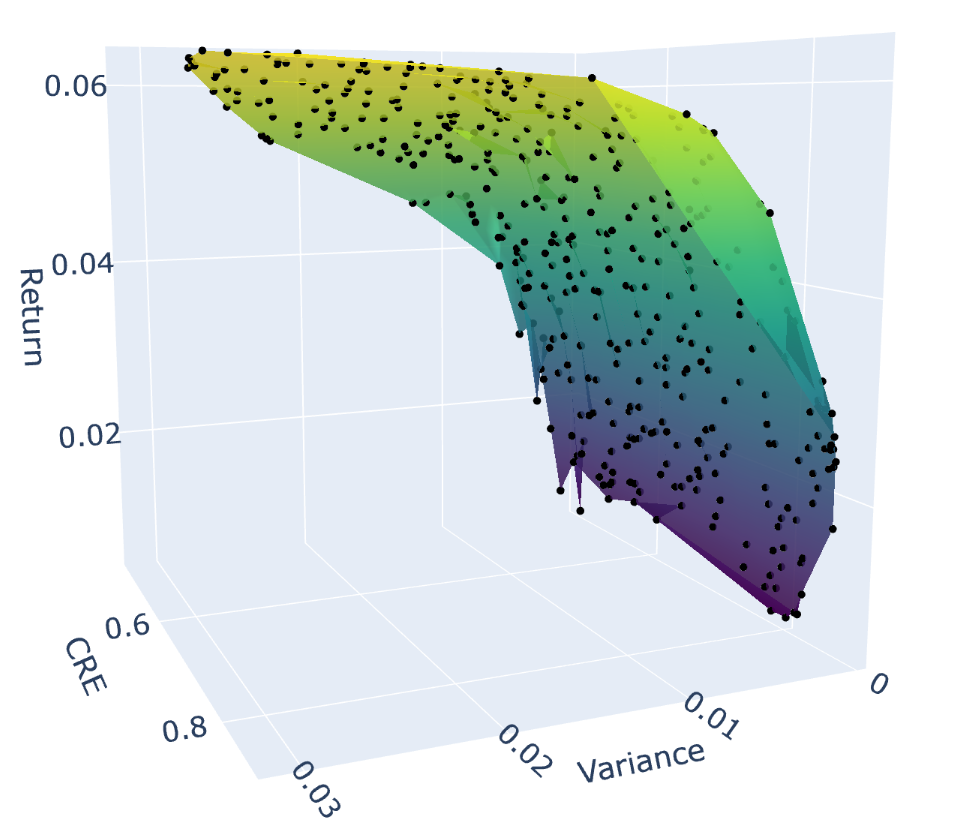}
        \caption{3D Pareto frontier integrating CRE}
        \label{fig:persp_1_cre}
    \end{subfigure}
    \hfill 
    \begin{subfigure}[b]{0.49\textwidth}
        \centering
        \includegraphics[width=\textwidth]{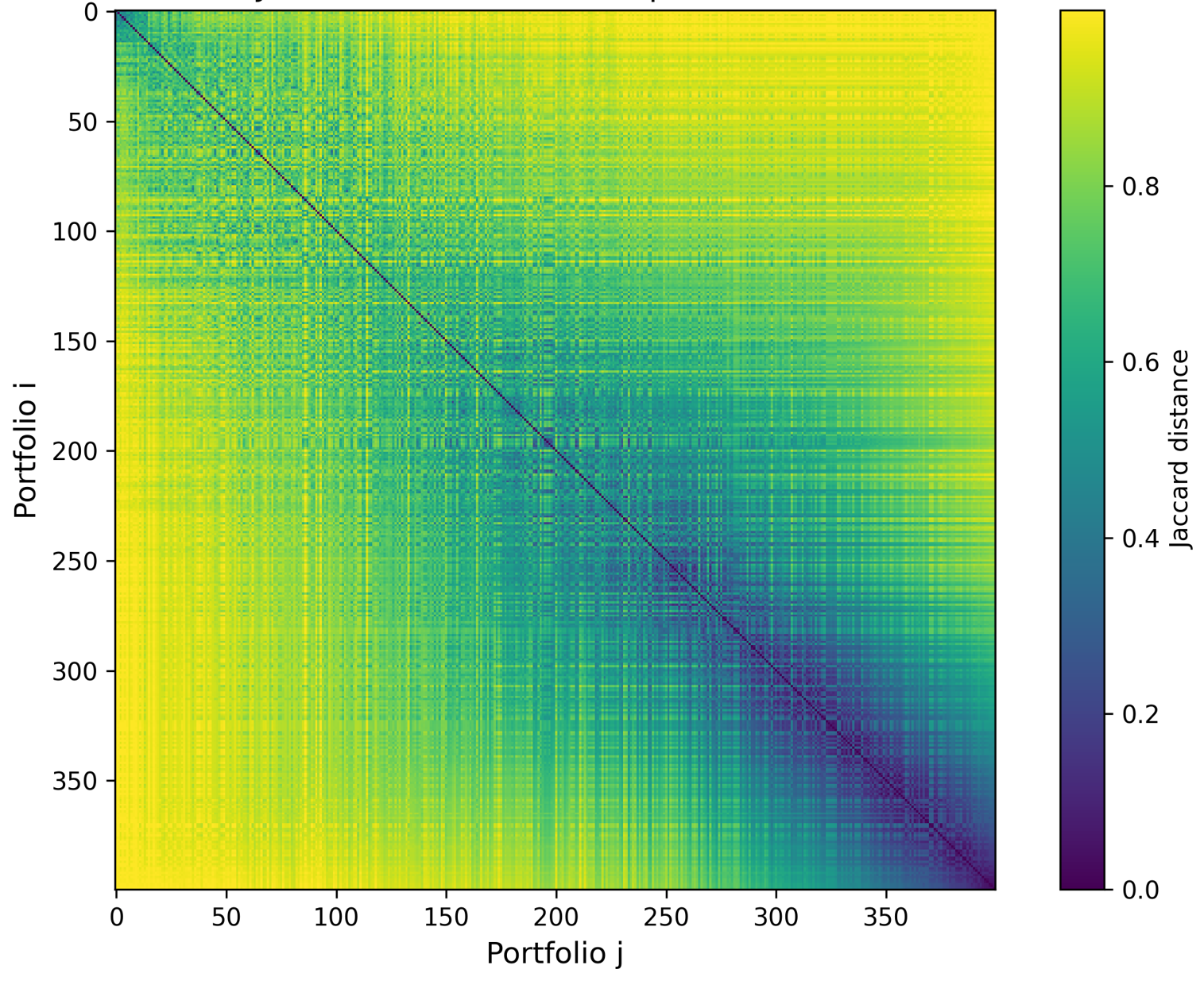}
        \caption{Jaccard distance heatmap for the CRE optimization}
        \label{fig:persp_2_cre}
    \end{subfigure}
    \label{fig:pareto_perspectives_all_cre}
\end{figure}

\vspace{0.7cm}

\begin{figure}[H]
    \centering
    \includegraphics[width=0.8\textwidth]{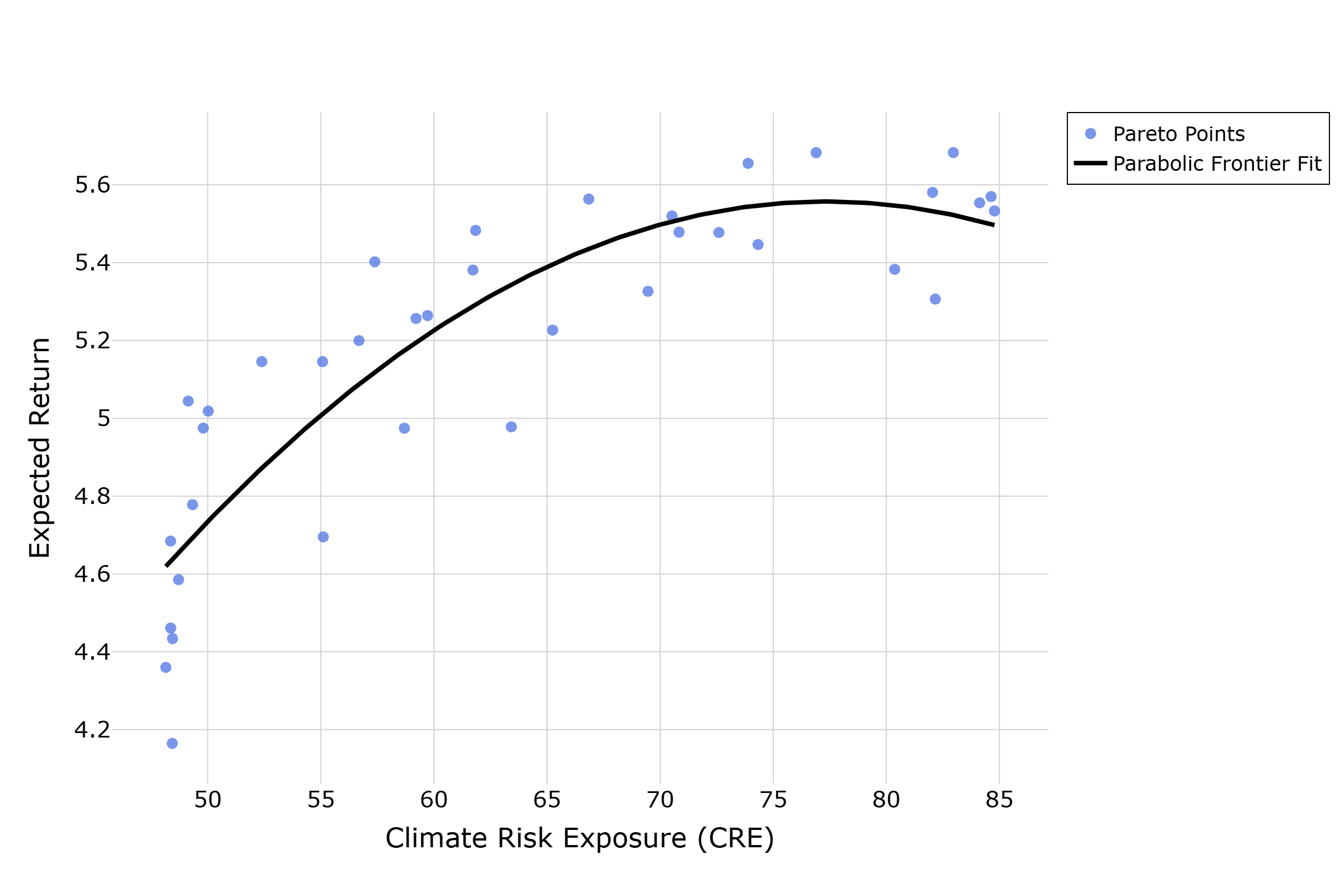} 
    \caption{Expected return vs. CRE: 2D cross-sectional Pareto projections (in percentage).}
    \label{fig:Cre_vs_Return}
\end{figure}

\vspace{0.7cm}
\section{Multi-objective optimization} \label{app:pareto_frontier_charts}

\subsection{The 3D Pareto frontier: exploring more perspectives}

\label{app:different_perspectives}
In this section of the appendix, additional viewpoints of the Pareto frontier, calculated at January 1, 2020, are presented to better visualize the multidimensional nature of the optimal solutions. Following the methodology described in Section \ref{sec:climate_opt}, this frontier is constructed within a multi-objective optimization framework designed to simultaneously minimize both market variance and \textit{Climate Exposure Volatility} (CEV), while maximizing the expected return. By observing these diverse perspectives, it is possible to appreciate the trade-offs between traditional financial risk, climate-related vulnerability, and potential growth across the set of optimal portfolios.

\begin{figure}[H]
    \centering
    \begin{subfigure}[b]{0.32\textwidth}
        \centering
        \includegraphics[width=\textwidth]{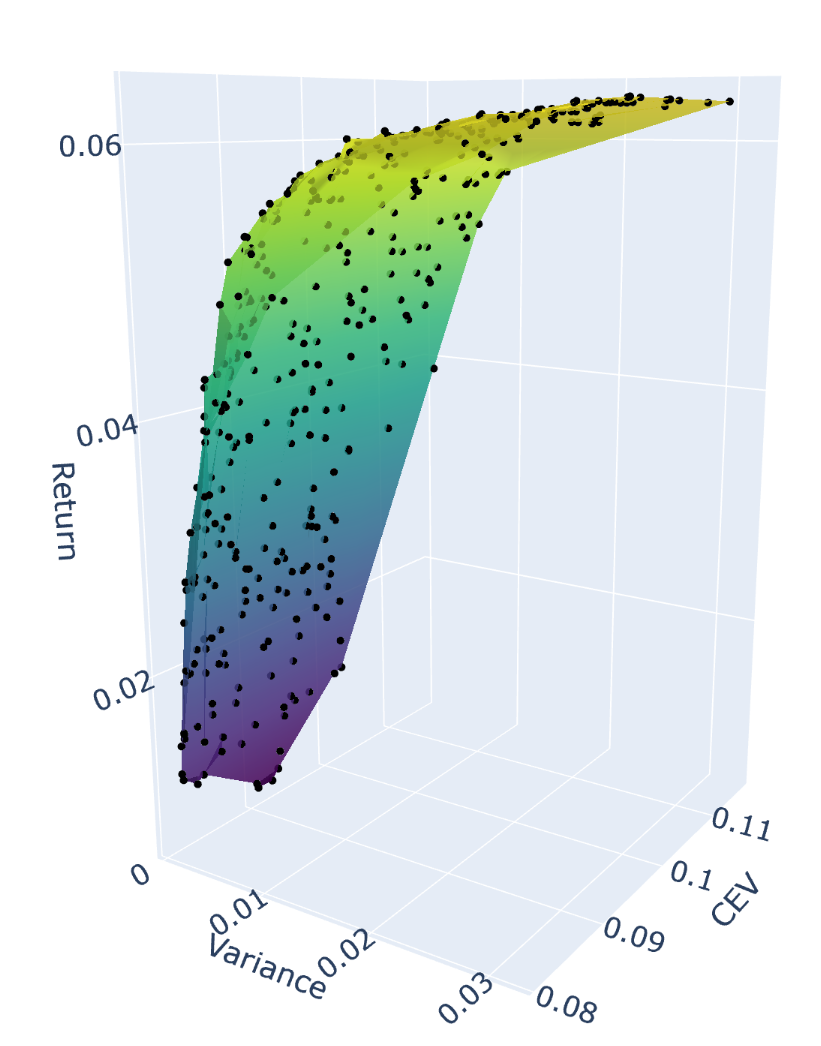}
        \caption{Perspective 1}
        \label{fig:persp_1}
    \end{subfigure}
    \hfill 
    \begin{subfigure}[b]{0.32\textwidth}
        \centering
        \includegraphics[width=\textwidth]{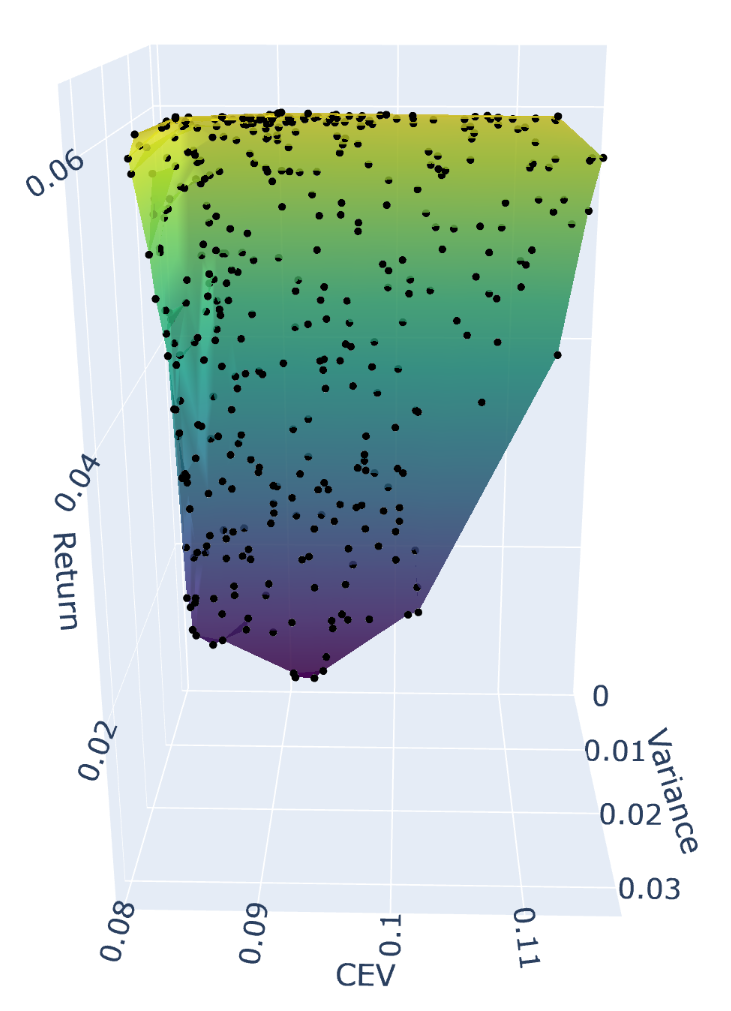}
        \caption{Perspective 2}
        \label{fig:persp_2}
    \end{subfigure}
    \hfill
    \begin{subfigure}[b]{0.28\textwidth}
        \centering
        \includegraphics[width=\textwidth]{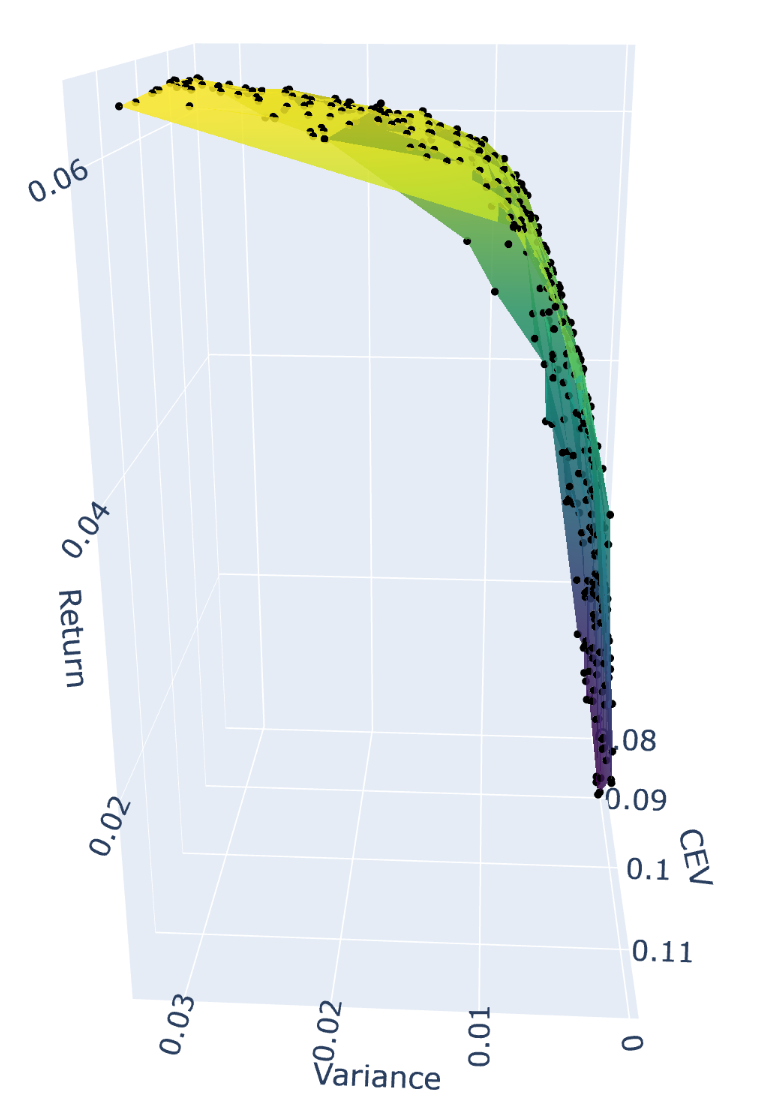}
        \caption{Perspective 3}
        \label{fig:persp_3}
    \end{subfigure}

    \label{fig:pareto_perspectives_all}
\end{figure}
\vspace{0.7cm}

\subsection{Multi-Objective Particle Swarm Optimization (MOPSO)}
\label{app:mopso}

The goal of this section is to provide additional details regarding the Multi-Objective Particle Swarm Optimization (MOPSO), the optimization algorithm employed to compute the Pareto fronts for the portfolio optimization problems discussed in this work. Specifically, this section provides comprehensive information regarding the algorithm's parameters and their functional roles, the specific values utilized in our simulations (see Table \ref{tab:mopso_parameters}), and a pseudocode illustrating the algorithm's execution flow (see Algorithm \ref{alg:mopso_standard}).




The scale of the search is defined by the population size ($N_{pop}$) and the number of iterations ($Iter$). While $N_{pop}$ represents the total number of candidate particles (portfolios) ``flying'' through the search space simultaneously, the iteration count determines the duration of the optimization process. Specifically, a higher $N_{pop}$ enhances the exploration of the search space, allowing the algorithm to evaluate a broader range of objective function values. On the other hand, a larger number of iterations ($Iter$) is crucial for the refinement of these solutions. To manage the non-dominated solutions discovered, the repository size ($N_{rep}$) limits the number of portfolios stored in the final archive, ensuring that the resulting Pareto front remains manageable yet representative of the entire trade-off spectrum.

To guarantee a well-distributed frontier, avoiding clusters of portfolios with nearly identical risk-return profiles, the algorithm employs a grid system ($N_{grid}$) and an inflation rate ($\alpha$). 
The grid system works by partitioning the objective space into hypercubes (cells). This discretization allows the algorithm to monitor the density of solutions across the frontier. By tracking how many portfolios fall into each cell, the MOPSO can steer the swarm toward less-populated regions, ensuring that the final output is not just a collection of optimal points, but a uniform representation of the entire risk-return trade-off.

In this context, the inflation rate ($\alpha$) is used to slightly expand the boundaries of the grid beyond the actual minimum and maximum objective values found in the repository ($s_{min}$ and $s_{max}$). Specifically, it calculates an offset $dc = (s_{max} - s_{min}) \cdot \alpha$, which is subtracted from the minimum and added to the maximum values.
Moreover, the selection pressure ($\beta$) and deletion pressure ($\gamma$) parameters are used to favor solutions in less-populated areas of the grid. 

The core of the MOPSO optimization process lies in the velocity update rule, which determines how each particle, representing a specific portfolio allocation, adjusts its movement through the search space. The direction and magnitude of a particle's displacement are governed by a stochastic velocity equation that balances three fundamental forces: inertia, social influence, and cognitive memory. Specifically, the velocity update combines the particle's previous velocity, modulated by the inertia weight $\omega$, with two acceleration components directed toward one of the global best positions found by the swarm ($gbest_{pos}$) and the individual particle's own historical best position ($pbest_{pos}$). 

The inertia component, $\omega$, provides the necessary momentum to maintain the particle's previous direction, ensuring stability in the search and preventing erratic changes that could lead to sub-optimal convergences. The social component, weighted by the coefficient $c_1$ and a random variable, pulls the particle toward a leader selected from the non-dominated repository. This selection, influenced by the pressure parameters $\beta$ and $\gamma$, ensures that the swarm gravitates toward the most promising and less explored regions of the current Pareto front. Simultaneously, the cognitive component, weighted by $c_2$, attracts the particle back to its personal best configuration. The interaction of these vectors, each scaled by random stochastic factors, defines both the trajectory and the intensity of the search. By modulating these weights, the algorithm effectively balances global exploration with local refinement, allowing the particles to adaptively shift the portfolio weights until the empirical frontier converges toward the theoretical optimum.

Finally, the mutation rate ($\mu$) introduces random variations in the positions of the particles, maintaining genetic diversity within the swarm and preventing the algorithm from converging prematurely on local sub-optimal frontiers.

\vspace{0.7cm}

\begin{table}[H]
\centering
\caption{MOPSO algorithm parameters and configuration.}
\label{tab:mopso_parameters}
\begin{tabular}{llc}
\hline
\textbf{Parameter} & \textbf{Description} & \textbf{Value} \\ \hline
$N_{pop}$ & Population size & 500 \\
$Iter$ & Number of iterations & 300 \\
$N_{rep}$ & External repository size & 400 \\
$N_{grid}$ & Number of grids per dimension & 10 \\
$\alpha$ & Grid inflation rate & 0.1 \\
$\beta$ & Selection pressure (global best selection) & 1 \\
$\gamma$ & Deletion pressure (repository maintenance) & 1 \\
$\mu$ & Mutation rate & 2 \\
$\omega$ & Inertia weight & 0.8 \\
$c_1$ & Personal learning coefficient & 1.5 \\
$c_2$ & Global learning coefficient & 1.5 \\
$lb$ & Lower bound of asset weights & 0 \\
$ub$ & Upper bound of asset weights & 1 \\ \hline
\end{tabular}
\end{table}

\begin{algorithm}[H]
\caption{Multi-Objective Particle Swarm Optimization (MOPSO)}
\label{alg:mopso_standard}
\begin{algorithmic}[1]

\STATE \textbf{Initialization:}
\STATE $pos \gets \text{list of } N_{pop} \text{ particles random initialized}$
\STATE \text{ensure particles' components are within bounds }$(lb, ub)$ and sum up to 1
\STATE $vel \gets \text{list of particle velocities initialized to 0}$
\STATE $objs \gets \text{list of objective values of the particles}$
\STATE $pbest \gets objs$, $pbest\_pos \gets pos $
\STATE $rep, rep\_obj \gets \text{non-dominated particles from } \{pos, objs\}$
\STATE $grid \gets \text{create\_grid}(rep\_obj, ngrid, \alpha)$

\STATE \textbf{Main loop:}
\FOR{$t = 0$ \TO $iter - 1$}
    \STATE $pm \gets (1 - t/(iter-1))^{1/\mu}$ \text{ (Mutation probability)}
    
    \FOR{$i = 0$ \TO $npop - 1$}
        \STATE \textbf{Update position:}
        \STATE $gbest\_pos \gets rep[\text{select\_gbest}(\beta)]$
        \STATE $vel[i] \gets \omega \, vel[i] + c_1 r_1 (gbest\_pos - pos[i]) + c_2 r_2 (pbest\_pos[i] - pos[i])$
        \STATE $pos[i] \gets \text{pos[i] + vel[i]}$
        \STATE $\text{ensure components of }pos[i] \text{ are within bounds } (ub,lb)$ and { sum up to 1}
        \STATE $objs[i] \gets \text{cal\_obj}(pos[i]) \text{ (update the objective values of the particle)}$

        \STATE \textbf{Mutation:}
        \IF{$\text{rand}() < pm$}
            \STATE $new\_pos \gets \text{mutate}(pos[i])$
            \STATE $new\_obj \gets \text{cal\_obj}(new\_pos)$
            \IF{$new\_obj$ dominates $objs[i]$ \OR ($objs[i]$ not dominates $new\_obj$ \AND $\text{rand}() < 0.5$)}
                \STATE $pos[i] \gets new\_pos$, $objs[i] \gets new\_obj$
                \STATE $\text{ensure components of }pos[i] \text{ are within bounds } (ub,lb)$  $ \text{and sum up to 1}$
            \ENDIF
        \ENDIF

        \STATE \textbf{Update personal best:}
        \IF{$objs[i]$ dominates $pbest[i]$ \OR ($pbest[i]$ not dominates $objs[i]$ \AND $\text{rand}() < 0.5$)}
            \STATE $pbest[i] \gets objs[i]$, $pbest\_pos[i] \gets pos[i]$
        \ENDIF
    \ENDFOR

    \STATE \textbf{Update repository and grid:}
    \STATE $rep.extend(pos)$, $rep\_obj.extend(objs)$
    \STATE remove dominated solutions from $rep$ and $rep\_obj$
    \STATE $grid \gets \text{create\_grid}(rep\_obj, ngrid, \alpha)$
    \WHILE{$length(rep) > N_{rep}$}
        \STATE $del\_ind \gets \text{select\_deletion}(\gamma)$
        \STATE $rep.pop(del\_ind)$, $rep\_obj.pop(del\_ind)$
    \ENDWHILE
\ENDFOR
\RETURN $result \gets \{rep, rep\_obj\}$
\end{algorithmic}
\end{algorithm}

\cleardoublepage


\section*{Acknowledgements}
We thank F. Cesarone, P. Manzoni, E. Sala, D. Stocco,  L. Viola and G. Zarfati and all participants to the Quantitative Finance Workshop 2026 for the insightful comments.
The research of G. Sbaiz was partially supported by the project of the Regional Programme (PR) FSE+ 2021/2027 of the Friuli Venezia Giulia Autonomous Region—PPO 2023—Specific Programme 22/23.  The research of M. Azzone is part of the activities of ``Dipartimento di Eccellenza 2023-2027''. M. Azzone and G. Sbaiz are part of the INdAM Research group GNAMPA.

\end{document}